\documentclass[11pt]{article}
\usepackage{xcpreprint}

\usepackage{amsmath,amssymb,amsthm}
\usepackage{graphicx}
\usepackage{booktabs}
\usepackage{longtable}
\usepackage{placeins}
\usepackage{xurl}
\usepackage{hyperref}
\usepackage[capitalise,noabbrev]{cleveref}

\makeatletter
\@ifclassloaded{elsarticle}{%
  \def\els@aparagraph[#1]#2{\elsparagraph[#1]{#2}}%
  \def\els@bparagraph#1{\elsparagraph*{#1}}%
}{}
\makeatother

\theoremstyle{plain}

\newtheorem{lemma}{Lemma}

\newcommand{\Var}{\operatorname{Var}}
\newcommand{\Cov}{\operatorname{Cov}}
\graphicspath{{figures/}}
\newcommand{\nBalFirst}{1,556}

\newcommand{\nFiveFirst}{1,428}
\newcommand{\nFiveLast}{3,284}

\newcommand{\aumBalFirst}{11.9}
\newcommand{\aumBalLast}{40.3}

\newcommand{\topTenFirst}{10.2}
\newcommand{\topTenLast}{22.9}

\newcommand{\nAct}{778}
\newcommand{\nPas}{778}
\newcommand{\eraAiOnlyWindow}{2023Q1--2025Q2}
\newcommand{\eraMiddleWindow}{2022Q4--2025Q2}
\newcommand{\eraLaterWindow}{2025Q3--2026Q2}

\newcommand{\hFiveVerdict}{not distinguishable from passive growth and concentration}

\newcommand{\hMoneBalPct}{+67.2}

\newcommand{\hMtwoBalPct}{+30.5}

\newcommand{\hMthreeActPct}{+211.3}

\newcommand{\hMthreexActPct}{-5.4}

\newcommand{\hMthreexPasPct}{-40.0}

\newcommand{\hTwoMtwoExPctBal}{+7.7}

\newcommand{\hTwoMoneExPctBal}{+12.7}

\newcommand{\hThreeVixBal}{0.283}
\newcommand{\hThreeVixPBal}{0.203}
\newcommand{\hThreeVixExQBal}{0.196}
\newcommand{\hThreeVixExQPBal}{0.435}

\newcommand{\doseLo}{0.010}
\newcommand{\doseHi}{0.026}
\newcommand{\doseLegacyCalibration}{0.36}
\newcommand{\phiReachCalibration}{0.17}
\newcommand{\phiReachCeiling}{0.083}
\newcommand{\doseBoundTrendFactor}{1.469}
\newcommand{\doseBoundCountCrossing}{129}

\newcommand{\doseMoondEarlyYear}{2019}
\newcommand{\doseMoondLateYear}{2024}
\newcommand{\doseMoondSample}{586}

\newcommand{\doseAimaFundManagers}{150}

\newcommand{\doseMoondLatePhiPct}{4.10}

\newcommand{\doseMoondEarlyPhiPct}{0.85}

\newcommand{\doseAimaPhiPct}{95.0}

\newcommand{\doseBrCountPredRho}{0.0138}

\newcommand{\hSixActMthreeP}{0.005}

\newcommand{\hSixActMthreexP}{0.218}

\newcommand{\eraNAgAct}{4}
\newcommand{\eraNMgrAct}{729}

\newcommand{\eraActMthreexCone}{0.065}

\newcommand{\eraActMthreexCtwo}{0.242}

\newcommand{\eraActMthreexDiff}{0.177}
\newcommand{\eraActMthreexDiffP}{0.023}

\newcommand{\hSixBalMthreeP}{2.52\times10^{-4}}

\newcommand{\hSixBalMthreexP}{0.327}

\newcommand{\eraBalMthreexDiff}{0.198}
\newcommand{\eraBalMthreexDiffP}{0.030}

\newcommand{\xMoondBalAggPct}{+227.7}

\newcommand{\xMoondBalAvgPct}{+27.6}

\newcommand{\xMoondFiveAggPct}{+421.3}

\newcommand{\xMoondFiveAvgPct}{-15.8}

\newcommand{\fpN}{50}

\newcommand{\fpNo}{29}
\newcommand{\fpRate}{58}

\newcommand{\nDthreeMissingAll}{4,951}
\newcommand{\nDthreeMissingScope}{4,944}
\newcommand{\nDthreeMissingManagers}{1,234}
\newcommand{\nDthreeInsufficient}{5,138}

\newcommand{\nCusipExcludedWindow}{212,803}

\newcommand{\nCusipPlaceholderCodesWindow}{440}

\newcommand{\figConvActExTopChange}{-0.0002}

\newcommand{\figConvPasExTopChange}{-0.009}

\newcommand{\owKappaP}{0.342}
\newcommand{\owKappaPLo}{0.242}
\newcommand{\owKappaPHi}{0.435}
\newcommand{\owAgreePPct}{67.7}
\newcommand{\owKappaPLabelOnly}{0.520}
\newcommand{\owKappaPLabelOnlyLo}{0.427}
\newcommand{\owKappaPLabelOnlyHi}{0.606}
\newcommand{\owInvalidP}{73}
\newcommand{\owPYes}{144}
\newcommand{\owPNo}{137}
\newcommand{\owPUnclear}{84}
\newcommand{\owRefYes}{168}
\newcommand{\owRefNo}{161}
\newcommand{\owRefUnclear}{36}
\newcommand{\owNScreen}{365}
\newcommand{\owKappaPS}{0.550}
\newcommand{\owKappaPSLo}{0.470}
\newcommand{\owKappaPSHi}{0.631}
\newcommand{\owPYesRefNot}{47}
\newcommand{\owPNotRefYes}{71}
\newcommand{\owNPortCoefP}{-0.2754}
\newcommand{\owNPortLoP}{-0.4650}
\newcommand{\owNPortHiP}{-0.0858}
\newcommand{\owBFiveExcludeCoefP}{-0.3058}
\newcommand{\owBFiveExcludeLoP}{-0.5173}
\newcommand{\owBFiveExcludeHiP}{-0.0944}
\newcommand{\owStressCovidDiffP}{0.1767}
\newcommand{\owStressCovidLoP}{-0.0055}
\newcommand{\owStressCovidHiP}{0.2295}
\newcommand{\owStressYenDiffP}{0.0131}
\newcommand{\owStressYenLoP}{-0.0275}
\newcommand{\owStressYenHiP}{0.0413}
\newcommand{\owStressTariffDiffP}{-0.0064}
\newcommand{\owStressTariffLoP}{-0.0271}
\newcommand{\owStressTariffHiP}{0.0398}
\newcommand{\owHFourPrimaryDiffPOnly}{0.0083}
\newcommand{\owHFourPrimaryLoPOnly}{0.0051}
\newcommand{\owHFourPrimaryHiPOnly}{0.0121}

\newcommand{\owHFourRobustDiffPOnly}{0.0055}
\newcommand{\owHFourRobustLoPOnly}{0.0022}
\newcommand{\owHFourRobustHiPOnly}{0.0081}

\newcommand{\owHFourPrimaryDiffP}{0.0086}
\newcommand{\owHFourPrimaryLoP}{0.0056}
\newcommand{\owHFourPrimaryHiP}{0.0126}
\newcommand{\owHFourPrimaryAiP}{131}
\newcommand{\owHFourPrimaryNoP}{780}
\newcommand{\owHFourRobustDiffP}{0.0056}
\newcommand{\owHFourRobustLoP}{0.0025}
\newcommand{\owHFourRobustHiP}{0.0081}

\newcommand{\owKappaS}{0.333}
\newcommand{\owKappaSLo}{0.251}
\newcommand{\owKappaSHi}{0.416}
\newcommand{\owAgreeSPct}{68.2}
\newcommand{\owKappaSLabelOnly}{0.451}
\newcommand{\owKappaSLabelOnlyLo}{0.359}
\newcommand{\owKappaSLabelOnlyHi}{0.536}
\newcommand{\owInvalidS}{100}
\newcommand{\owSYes}{76}
\newcommand{\owSNo}{186}
\newcommand{\owSUnclear}{103}
\newcommand{\owNPortCoefS}{-0.1686}
\newcommand{\owNPortLoS}{-0.3486}
\newcommand{\owNPortHiS}{0.0114}
\newcommand{\owBFiveExcludeCoefS}{-0.2019}
\newcommand{\owBFiveExcludeLoS}{-0.3743}
\newcommand{\owBFiveExcludeHiS}{-0.0296}

\newcommand{\owHFourPrimaryDiffS}{0.0071}
\newcommand{\owHFourPrimaryLoS}{0.0033}
\newcommand{\owHFourPrimaryHiS}{0.0119}

\newcommand{\owHFourRobustDiffS}{0.0073}
\newcommand{\owHFourRobustLoS}{0.0049}
\newcommand{\owHFourRobustHiS}{0.0098}

\newcommand{\owHFourPrimaryDiffRef}{-0.0011}
\newcommand{\owHFourPrimaryLoRef}{-0.0054}
\newcommand{\owHFourPrimaryHiRef}{0.0038}

\newcommand{\owHFourRobustDiffRef}{-0.0004}
\newcommand{\owHFourRobustLoRef}{-0.0047}
\newcommand{\owHFourRobustHiRef}{0.0046}

\newcommand{\nPortSeries}{17,977}
\newcommand{\nPortFilings}{331,101}
\newcommand{\nPortReturns}{987,731}
\newcommand{\nPortMedianClass}{526,938}

\newcommand{\nPortValidLabels}{2,533}

\newcommand{\nPortCrdSeriesPct}{91.2}
\newcommand{\nPortCrdAumPct}{97.2}

\newcommand{\nPortAdvAumPct}{96.0}
\newcommand{\nPortValidLabelsPct}{14.1}
\newcommand{\nPortValidLabelsAumPct}{10.6}
\newcommand{\nPortComovementCoef}{-0.1397}
\newcommand{\nPortComovementLo}{-0.3166}
\newcommand{\nPortComovementHi}{0.0372}
\newcommand{\nPortComovementMde}{0.2510}

\newcommand{\nPortComovementDf}{33}

\newcommand{\nPortComovementFundLo}{-0.2026}
\newcommand{\nPortComovementFundHi}{-0.0768}

\newcommand{\nPortComovementPairWindows}{6,648,301}
\newcommand{\nPortComovementAiPairs}{62,564}
\newcommand{\nPortComovementAiPairWindows}{763,710}
\newcommand{\nPortComovementAiFunds}{601}
\newcommand{\nPortComovementNoFunds}{1,422}
\newcommand{\nPortComovementAiAdvisers}{34}
\newcommand{\nPortComovementNoAdvisers}{168}
\newcommand{\nPortComovementAiFundsPerAdviser}{17.7}
\newcommand{\nPortComovementNoFundsPerAdviser}{8.5}

\newcommand{\nPortStressCovidDiff}{0.0158}
\newcommand{\nPortStressCovidLo}{-0.0101}
\newcommand{\nPortStressCovidHi}{0.0679}
\newcommand{\nPortStressCovidMde}{0.0588}
\newcommand{\nPortStressCovidBootstrapSE}{0.0210}
\newcommand{\nPortStressCovidBootstrapDraws}{1,000}
\newcommand{\nPortStressCovidAiFunds}{183}
\newcommand{\nPortStressCovidNoFunds}{316}
\newcommand{\nPortStressCovidAiAdvisers}{11}
\newcommand{\nPortStressCovidNoAdvisers}{38}

\newcommand{\nPortStressCovidPowerFloor}{met}

\newcommand{\nPortStressYenDiff}{0.0411}
\newcommand{\nPortStressYenLo}{-0.0034}
\newcommand{\nPortStressYenHi}{0.0679}
\newcommand{\nPortStressYenMde}{0.0601}
\newcommand{\nPortStressYenBootstrapSE}{0.0214}

\newcommand{\nPortStressYenAiFunds}{333}
\newcommand{\nPortStressYenNoFunds}{865}
\newcommand{\nPortStressYenAiAdvisers}{26}
\newcommand{\nPortStressYenNoAdvisers}{118}

\newcommand{\nPortStressYenPowerFloor}{met}

\newcommand{\nPortStressTariffDiff}{0.0332}
\newcommand{\nPortStressTariffLo}{-0.0156}
\newcommand{\nPortStressTariffHi}{0.0629}
\newcommand{\nPortStressTariffMde}{0.0651}
\newcommand{\nPortStressTariffBootstrapSE}{0.0233}

\newcommand{\nPortStressTariffAiFunds}{338}
\newcommand{\nPortStressTariffNoFunds}{820}
\newcommand{\nPortStressTariffAiAdvisers}{26}
\newcommand{\nPortStressTariffNoAdvisers}{119}

\newcommand{\nPortStressTariffPowerFloor}{met}

\newcommand{\advBOneTwoLabels}{1,067}
\newcommand{\advBOneTwoReference}{50}

\newcommand{\advBOneTwoTP}{24}
\newcommand{\advBOneTwoFP}{5}
\newcommand{\advBOneTwoFN}{0}
\newcommand{\advBOneTwoTN}{21}

\newcommand{\advBOneTwoOriginalDisagreements}{10}
\newcommand{\advBOneTwoAuthorAdjudicated}{10}
\newcommand{\advBOneTwoYes}{163}
\newcommand{\advBOneTwoNo}{868}
\newcommand{\advBOneTwoUnclear}{36}
\newcommand{\advBOneTwoPrecision}{0.83}

\newcommand{\advBOneTwoRecall}{1.00}

\newcommand{\advBOneTwoFOne}{0.91}

\newcommand{\advBOneTwoSpecificity}{0.81}

\newcommand{\advBOneTwoAccuracy}{0.90}

\newcommand{\advBOneTwoPrecisionLo}{0.642}
\newcommand{\advBOneTwoPrecisionHi}{0.942}
\newcommand{\advBOneTwoRecallLo}{0.858}
\newcommand{\advBOneTwoRecallHi}{1.000}

\newcommand{\advBOneTwoOriginalAgreements}{40}
\newcommand{\hFourBOneTwoCoveragePct}{27.5}

\newcommand{\hFourBOneTwoPrimaryAi}{151}
\newcommand{\hFourBOneTwoPrimaryNo}{802}

\newcommand{\hFourBOneTwoPrimaryChAi}{0.0230}
\newcommand{\hFourBOneTwoPrimaryChNon}{0.0251}
\newcommand{\hFourBOneTwoPrimaryDiff}{-0.0022}
\newcommand{\hFourBOneTwoPrimaryLo}{-0.0066}
\newcommand{\hFourBOneTwoPrimaryHi}{0.0029}
\newcommand{\hFourBOneTwoPrimaryPreQuarters}{15}
\newcommand{\hFourBOneTwoPrimaryPostQuarters}{10}
\newcommand{\hFourBOneTwoPrimaryBootstrapDraws}{1,000}
\newcommand{\hFourBOneTwoPrimaryBlockLength}{4}
\newcommand{\hFourBOneTwoPrimaryAumPct}{25.9}

\newcommand{\hFourBOneTwoRobustAi}{184}
\newcommand{\hFourBOneTwoRobustNo}{802}

\newcommand{\hFourBOneTwoRobustChAi}{0.0237}
\newcommand{\hFourBOneTwoRobustChNon}{0.0251}
\newcommand{\hFourBOneTwoRobustDiff}{-0.0014}
\newcommand{\hFourBOneTwoRobustLo}{-0.0059}
\newcommand{\hFourBOneTwoRobustHi}{0.0038}

\newcommand{\hFourBOneTwoRobustAumPct}{27.5}

\newcommand{\nPortBThreeAgreementCoef}{-0.0649}
\newcommand{\nPortBThreeAgreementLo}{-0.1410}
\newcommand{\nPortBThreeAgreementHi}{0.0111}
\newcommand{\nPortBThreeAgreementMde}{0.0999}
\newcommand{\nPortBThreeAgreementSE}{0.0177}

\newcommand{\nPortBThreeAgreementPairWindows}{34,682}

\newcommand{\nPortBThreeAgreementAiPairWindows}{1,954}
\newcommand{\nPortBThreeAgreementAiFunds}{45}
\newcommand{\nPortBThreeAgreementNoFunds}{134}
\newcommand{\nPortBThreeAgreementAiAdvisers}{3}
\newcommand{\nPortBThreeAgreementNoAdvisers}{6}

\newcommand{\nPortBThreeAgreementPowerFloor}{met}

\newcommand{\nPortBThreeClassifierCoef}{-0.1397}
\newcommand{\nPortBThreeClassifierLo}{-0.3166}
\newcommand{\nPortBThreeClassifierHi}{0.0372}
\newcommand{\nPortBThreeClassifierMde}{0.2510}
\newcommand{\nPortBThreeClassifierSE}{0.0869}

\newcommand{\nPortBThreeClassifierPairWindows}{6,648,301}

\newcommand{\nPortBThreeClassifierAiPairWindows}{763,710}
\newcommand{\nPortBThreeClassifierAiFunds}{601}
\newcommand{\nPortBThreeClassifierNoFunds}{1,422}
\newcommand{\nPortBThreeClassifierAiAdvisers}{34}
\newcommand{\nPortBThreeClassifierNoAdvisers}{168}

\newcommand{\nPortBThreeClassifierPowerFloor}{met}

\newcommand{\nPortBThreeManualCoef}{-0.0431}
\newcommand{\nPortBThreeManualLo}{-0.2286}
\newcommand{\nPortBThreeManualHi}{0.1423}
\newcommand{\nPortBThreeManualMde}{0.2631}
\newcommand{\nPortBThreeManualSE}{0.0912}

\newcommand{\nPortBThreeManualPairWindows}{6,321,632}

\newcommand{\nPortBThreeManualAiPairWindows}{1,024,107}
\newcommand{\nPortBThreeManualAiFunds}{682}
\newcommand{\nPortBThreeManualNoFunds}{1,341}
\newcommand{\nPortBThreeManualAiAdvisers}{35}
\newcommand{\nPortBThreeManualNoAdvisers}{167}

\newcommand{\nPortBThreeManualPowerFloor}{met}

\newcommand{\nPortBFiveExcludeCoef}{-0.1667}
\newcommand{\nPortBFiveExcludeLo}{-0.3396}
\newcommand{\nPortBFiveExcludeHi}{0.0061}
\newcommand{\nPortBFiveExcludeMde}{0.2452}
\newcommand{\nPortBFiveExcludeSE}{0.0848}

\newcommand{\nPortBFiveExcludePairWindows}{4,049,461}

\newcommand{\nPortBFiveExcludeAiPairWindows}{662,640}
\newcommand{\nPortBFiveExcludeAiFunds}{528}
\newcommand{\nPortBFiveExcludeNoFunds}{1,024}
\newcommand{\nPortBFiveExcludeAiAdvisers}{33}
\newcommand{\nPortBFiveExcludeNoAdvisers}{166}

\newcommand{\nPortBFiveExcludePowerFloor}{met}

\newcommand{\nPortBFiveControlsCoef}{-0.1597}
\newcommand{\nPortBFiveControlsLo}{-0.3325}
\newcommand{\nPortBFiveControlsHi}{0.0132}
\newcommand{\nPortBFiveControlsMde}{0.2452}
\newcommand{\nPortBFiveControlsSE}{0.0850}

\newcommand{\nPortBFiveControlsPairWindows}{6,648,301}

\newcommand{\nPortBFiveControlsAiPairWindows}{763,710}
\newcommand{\nPortBFiveControlsAiFunds}{601}
\newcommand{\nPortBFiveControlsNoFunds}{1,422}
\newcommand{\nPortBFiveControlsAiAdvisers}{34}
\newcommand{\nPortBFiveControlsNoAdvisers}{168}

\newcommand{\nPortBFiveControlsPowerFloor}{met}

\newcommand{\nPortBFiveExcludeClassifierCoef}{-0.1667}
\newcommand{\nPortBFiveExcludeClassifierLo}{-0.3396}
\newcommand{\nPortBFiveExcludeClassifierHi}{0.0061}
\newcommand{\nPortBFiveExcludeClassifierMde}{0.2452}
\newcommand{\nPortBFiveExcludeClassifierSE}{0.0848}

\newcommand{\nPortBFiveExcludeClassifierPairWindows}{4,049,461}

\newcommand{\nPortBFiveExcludeClassifierAiPairWindows}{662,640}
\newcommand{\nPortBFiveExcludeClassifierAiFunds}{528}
\newcommand{\nPortBFiveExcludeClassifierNoFunds}{1,024}
\newcommand{\nPortBFiveExcludeClassifierAiAdvisers}{33}
\newcommand{\nPortBFiveExcludeClassifierNoAdvisers}{166}

\newcommand{\nPortBFiveExcludeClassifierPowerFloor}{met}

\newcommand{\nPortBFiveExcludeManualCoef}{-0.1205}
\newcommand{\nPortBFiveExcludeManualLo}{-0.2916}
\newcommand{\nPortBFiveExcludeManualHi}{0.0506}
\newcommand{\nPortBFiveExcludeManualMde}{0.2428}
\newcommand{\nPortBFiveExcludeManualSE}{0.0841}

\newcommand{\nPortBFiveExcludeManualPairWindows}{3,968,633}

\newcommand{\nPortBFiveExcludeManualAiPairWindows}{759,132}
\newcommand{\nPortBFiveExcludeManualAiFunds}{561}
\newcommand{\nPortBFiveExcludeManualNoFunds}{991}
\newcommand{\nPortBFiveExcludeManualAiAdvisers}{34}
\newcommand{\nPortBFiveExcludeManualNoAdvisers}{165}

\newcommand{\nPortBFiveExcludeManualPowerFloor}{met}

\newcommand{\nPortBSixExpertTargetSeries}{0}

\newcommand{\nPortATwoPrimaryGamma}{0.0362}
\newcommand{\nPortATwoPrimarySE}{0.1394}
\newcommand{\nPortATwoPrimaryMde}{0.3905}
\newcommand{\nPortATwoPrimaryLo}{-0.2369}
\newcommand{\nPortATwoPrimaryHi}{0.3094}
\newcommand{\nPortATwoPrimaryPreDiff}{-0.1726}
\newcommand{\nPortATwoPrimaryPreLo}{-0.4485}
\newcommand{\nPortATwoPrimaryPreHi}{0.1033}
\newcommand{\nPortATwoPrimaryPostDiff}{-0.1364}
\newcommand{\nPortATwoPrimaryPostLo}{-0.3715}
\newcommand{\nPortATwoPrimaryPostHi}{0.0988}
\newcommand{\nPortATwoPrimaryAiFunds}{276}
\newcommand{\nPortATwoPrimaryNoFunds}{543}
\newcommand{\nPortATwoPrimaryAiAdvisers}{24}
\newcommand{\nPortATwoPrimaryNoAdvisers}{121}
\newcommand{\nPortATwoPrimaryPowerFloor}{met}
\newcommand{\nPortATwoPrimaryVerdict}{uninformative}
\newcommand{\nPortATwoUnclearGamma}{0.0046}
\newcommand{\nPortATwoUnclearSE}{0.1169}
\newcommand{\nPortATwoUnclearMde}{0.3276}
\newcommand{\nPortATwoUnclearLo}{-0.2246}
\newcommand{\nPortATwoUnclearHi}{0.2338}
\newcommand{\nPortATwoUnclearPreDiff}{-0.0697}
\newcommand{\nPortATwoUnclearPreLo}{-0.3928}
\newcommand{\nPortATwoUnclearPreHi}{0.2533}
\newcommand{\nPortATwoUnclearPostDiff}{-0.0652}
\newcommand{\nPortATwoUnclearPostLo}{-0.3110}
\newcommand{\nPortATwoUnclearPostHi}{0.1807}
\newcommand{\nPortATwoUnclearAiFunds}{377}
\newcommand{\nPortATwoUnclearNoFunds}{543}
\newcommand{\nPortATwoUnclearAiAdvisers}{31}
\newcommand{\nPortATwoUnclearNoAdvisers}{121}
\newcommand{\nPortATwoUnclearPowerFloor}{met}
\newcommand{\nPortATwoUnclearVerdict}{uninformative}
\newcommand{\nPortATwoIndexGamma}{0.0326}
\newcommand{\nPortATwoIndexSE}{0.1390}
\newcommand{\nPortATwoIndexMde}{0.3895}
\newcommand{\nPortATwoIndexLo}{-0.2399}
\newcommand{\nPortATwoIndexHi}{0.3051}
\newcommand{\nPortATwoIndexPreDiff}{-0.3867}
\newcommand{\nPortATwoIndexPreLo}{-0.7251}
\newcommand{\nPortATwoIndexPreHi}{-0.0483}
\newcommand{\nPortATwoIndexPostDiff}{-0.3541}
\newcommand{\nPortATwoIndexPostLo}{-0.6746}
\newcommand{\nPortATwoIndexPostHi}{-0.0336}
\newcommand{\nPortATwoIndexAiFunds}{19}
\newcommand{\nPortATwoIndexNoFunds}{192}
\newcommand{\nPortATwoIndexAiAdvisers}{3}
\newcommand{\nPortATwoIndexNoAdvisers}{9}
\newcommand{\nPortATwoIndexPowerFloor}{not met}
\newcommand{\nPortATwoIndexVerdict}{uninformative}

\newcommand{\propAddedExitShift}{0.0223}
\newcommand{\propAddedExitShiftLimit}{0.0443}
\newcommand{\propStaticLoopWidth}{2.009}
\newcommand{\propAddedExitSharePct}{1.1}
\newcommand{\propAddedExitShareBoundPct}{2.204}
\newcommand{\propAddedExitShareBoundRoundedPct}{2.2}
\newcommand{\propEntryFoldShare}{0.664}
\newcommand{\propExitFoldShare}{0.971}

\newcommand{\propTerminalLoopWidth}{1.967}
\newcommand{\propSameStateNarrowPct}{2.07}
\newcommand{\propInitialVariance}{0.1}

\newcommand{\propParamSigmaV}{1.0}
\newcommand{\propParamSigmaU}{1.0}
\newcommand{\propParamSigmaEta}{0.3}
\newcommand{\propParamRho}{0.15}
\newcommand{\propParamBeta}{0.2}
\newcommand{\propParamN}{1000}
\newcommand{\propParamTau}{1.0}
\newcommand{\propParamGamma}{0.2}
\newcommand{\propParamChiH}{1.0}
\newcommand{\propParamVarHZero}{0.0}
\newcommand{\propParamVarAi}{0.2}
\newcommand{\propParamDelta}{0.6}
\newcommand{\propMixWeight}{0.65}
\newcommand{\propMixLocationLow}{-6.0}
\newcommand{\propMixLocationHigh}{0.6}
\newcommand{\propMixScale}{0.04}
\newcommand{\propDepreciation}{0.1}
\newcommand{\propHorizon}{12}

\newcommand{\advScreenHit}{365}
\newcommand{\advScreenNoHit}{702}
\newcommand{\advScreenNoHitSharePct}{65.8}
\newcommand{\advScreenRuleNoSharePct}{81.3}
\newcommand{\advScreenAuditN}{50}
\newcommand{\advScreenAuditFrameN}{702}
\newcommand{\advScreenAuditFalseNegatives}{0}
\newcommand{\advScreenAuditUpperPct}{5.8}

\newcommand{\unitNoConsensusPrimaryReports}{17}
\newcommand{\unitNoConsensusPostSwitchReports}{4}
\newcommand{\unitNoConsensusSfiveReports}{6}
\newcommand{\unitNoConsensusBalancedReports}{5}
\newcommand{\unitNoConsensusDOneMissing}{15}
\newcommand{\unitNoConsensusDOneAvailable}{2}

\newcommand{\unitNoConsensusAdditionalDthree}{0}
\newcommand{\unitNoConsensusPrimaryBalancedN}{1,556}
\newcommand{\unitNoConsensusDropBalancedN}{1,551}
\newcommand{\unitNoConsensusMonePrimaryPct}{+67.2}
\newcommand{\unitNoConsensusMoneDropPct}{+67.0}
\newcommand{\unitNoConsensusMtwoPrimaryPct}{+30.5}
\newcommand{\unitNoConsensusMtwoDropPct}{+30.5}
\newcommand{\unitNoConsensusMthreexPrimaryPct}{-5.4}
\newcommand{\unitNoConsensusMthreexDropPct}{-5.2}

\begin{document}

\xctitleblock{AI Disclosures and Fund Return Comovement: Evidence from Public Filings}{Shuchen Meng\textsuperscript{1} \quad Xupeng Chen\textsuperscript{2,*}}{\textsuperscript{1}Department of Financial Engineering, New York University\\\textsuperscript{2}Department of Electrical and Computer Engineering, New York University\\\textsuperscript{*}Corresponding author. Email: \href{mailto:xc1490@nyu.edu}{xc1490@nyu.edu}}{This version: November 2026. It replaces an earlier version of this arXiv entry; see the Comments field.}
\begin{abstract}
Shared AI systems could make institutions respond to information in similar ways and sell the same assets under stress. We use public filings to ask whether advisers that disclose AI use are associated with greater fund-return comovement, greater sales overlap in stress windows, or portfolio convergence that survives benchmark controls. We link current adviser disclosures to monthly Form N-PORT returns and holdings, and use Form 13F holdings to separate benchmark and concentration effects. Under the frozen disclosure labels, the N-PORT return-comovement estimate is negative and its 95 per cent interval includes zero, while each stress-window sales-overlap estimate is positive and imprecise. When applying the same rule, Qwen3-235B agrees weakly with the original classifier (Cohen's kappa \owKappaP{}). With its labels, the negative comovement estimate has a 95 per cent interval that excludes zero. In an exploratory sales-overlap comparison, one point estimate changes sign, but its interval still includes zero. In the balanced 13F panel, changing from the aggregate to the average-manager benchmark reduces the fitted active-position increase from \xMoondBalAggPct{} to \xMoondBalAvgPct{} per cent. The sign reverses only in the changing \$500 million sample. The evidence describes disclosure groups and does not identify AI adoption or common model use. Supervisory data that record deployment timing, model provenance, decision use and human overrides would permit a direct test of common-model fragility.
\end{abstract}
\xckeywords{G01, G14, G23, G28}{artificial intelligence; systemic risk; herding; hysteresis; institutional investors}

\section{Introduction}
\label{sec:intro}

Common AI systems could create correlated investment decisions and common selling during stress. The Financial Stability Board identifies reliance on common providers and market correlations as potential sources of financial vulnerability \citep{FSB2024,IMF2024,ECB2024}. Public filings do not reveal signals or orders, so the relevant question is narrower. We ask whether observable disclosure groups differ in return comovement, stress-period sales overlap, or measured holdings convergence after accounting for common holdings and benchmark choice.

\begin{figure}[t]
\centering
\includegraphics[width=\textwidth]{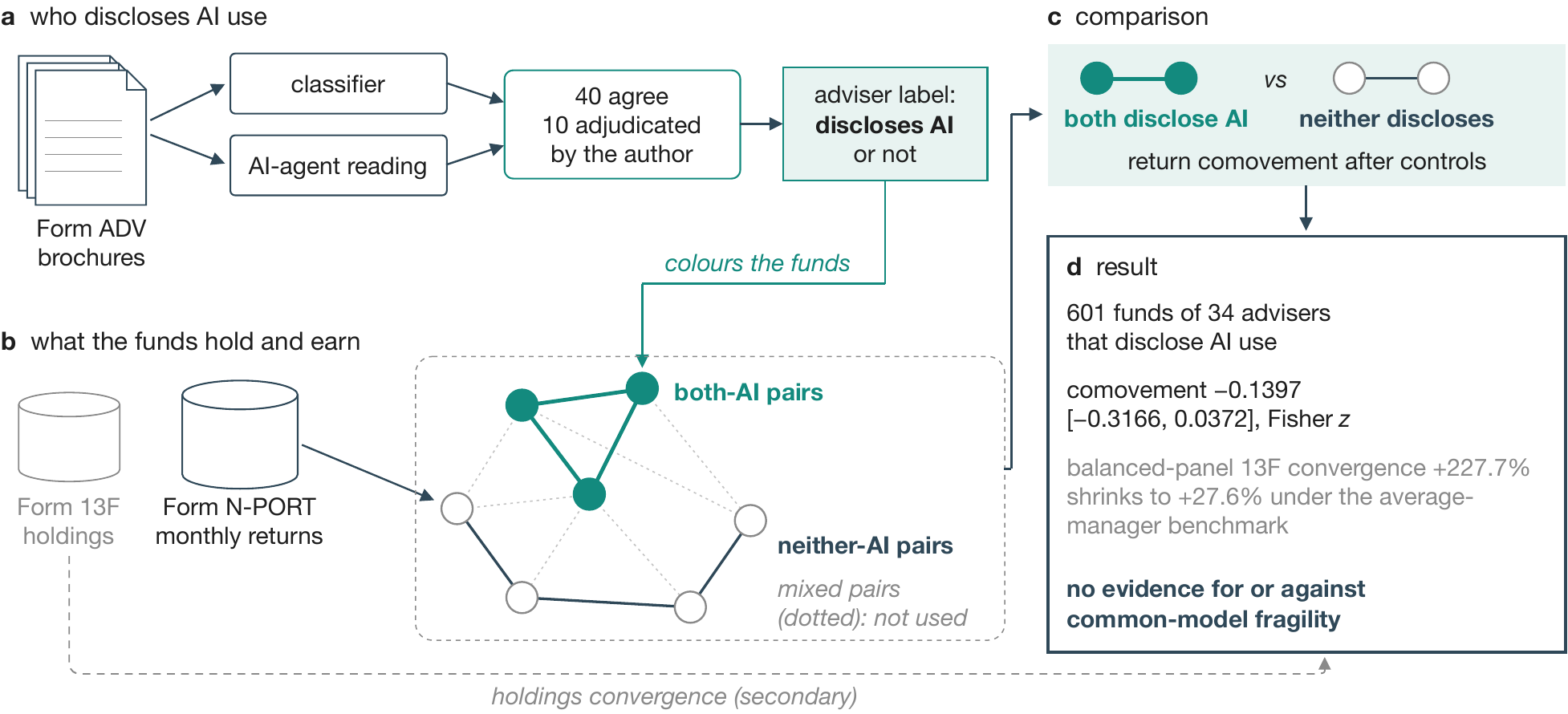}
\caption{\textbf{Study design for disclosure-group comparisons.} Adviser labels are derived from Form ADV brochures and link to funds in monthly Form N-PORT data. The return comparison uses pairs whose advisers both disclose AI use and pairs whose advisers are both classified as having no detected disclosure; mixed pairs are excluded. The outcome is conditional fund-return comovement after holdings overlap, fund size and month controls. The figure is schematic and does not imply that a no-detected-disclosure label means that an adviser does not use AI.}
\label{fig:design}
\end{figure}

The paper separates return comovement from common orders and signals. \citet{GreenwoodThesmar2011} link ownership concentration and correlated investor flows to volatility and return comovement, while \citet{AntonPolk2014} show that common active-fund ownership predicts excess stock-return comovement. Styles, asset classes, exposure sizes, hedges and unobserved positions can still separate returns from orders. Holdings similarity can also rise mechanically with market-capitalisation concentration or passive growth. Our Form 13F evidence makes that measurement issue visible because the balanced-panel increase is much smaller under the average-manager benchmark and is concentrated in positions that include the ten largest stocks.

The contribution is a public-data measurement comparison. \citet{ChenSialmXu2026} use archived disclosures and find lower benchmark-adjusted return correlations among disclosed-AI hedge funds. We study a broader monthly fund panel, condition on holdings overlap and fund size, and add stress-window sales overlap. \citet{Sheng2026} infer generative-AI reliance from the explanatory power of generated signals for holding changes and validate that measure by survey. Our measure records current adviser disclosure and cannot date adoption. \citet{Moond2026} shows that 13F convergence depends on the benchmark and the largest stocks, including an exact top-ten-stock decomposition. We carry the benchmark comparison into the balanced and changing \$500 million samples; this benchmark comparison is distinct from their stock-level decomposition. None of these studies observes common signals or order covariance directly.

We classify \advBOneTwoLabels{} adviser brochures, link current labels to N-PORT series, and compare both-disclosure pairs with neither-disclosure pairs. We estimate conditional return comovement and stress-window sales overlap, then use a balanced 13F panel of \nBalFirst{} managers to examine benchmark and top-ten concentration. The comparisons are descriptive. The filings do not date adoption or systematically identify deployed model versions, shared data inputs or the decisions in which outputs were used.

The accompanying mechanism is interpretive rather than empirical. It distinguishes common signals from incentives to adopt and retain a system, and its assumptions and uncalibrated illustration are reported in the appendix. It does not identify a mechanism in the filing data.

\section{Data and disclosure classification}
\label{sec:data}

\paragraph{N-PORT and adviser links.} The N-PORT panel contains \nPortSeries{} series, \nPortFilings{} original filings and \nPortReturns{} unique series-month return observations from July 2019 through July 2026. The final panel keeps the latest original filing for each series and report date, excludes amendments, and keeps the latest filing when quarterly releases overlap for a series-month. Holdings controls retain positive-value, valid-CUSIP, long non-derivative common and preferred equity. The stress-window outcome uses common equity only. The security join requires a valid full CUSIP before aggregation. Optional ISIN and issuer-LEI fields are audited but are never used as a security join or fallback key. Time-valid exact Central Registration Depository (CRD) links cover \nPortCrdSeriesPct{} per cent of series and \nPortCrdAumPct{} per cent of latest-snapshot net assets. Form ADV bulk data match \nPortAdvAumPct{} per cent of assets. Because the public monthly-return table does not report class net assets, \nPortMedianClass{} multiple-class observations use the median-class fallback. Current-vintage labels are available for \nPortValidLabels{} series, representing \nPortValidLabelsPct{} per cent of series and \nPortValidLabelsAumPct{} per cent of latest-snapshot assets. These links do not identify when an adviser adopted AI. A series is flagged as index or ETF by the union of Form N-CEN's \texttt{IS\_INDEX=Y} and \texttt{IS\_ETF=Y} items and a literal ``Index'' match in the latest eligible N-CEN fund name or N-PORT series name.

\paragraph{Disclosure classification and validation.} We classify \advBOneTwoLabels{} adviser brochures according to whether an adviser states that it uses AI or machine learning in investment research or decisions. The literal screen covers AI/ML variants, neural and deep learning, language models and natural-language processing, generative/GPT terms and data science. It does not include the term \emph{expert system}; a generic quantitative, algorithmic or computer model is insufficient for the separate expert-system boundary classification. For \advScreenHit{} advisers with a hit, the classifier reads the hit sentence and one adjacent sentence on either side. The other \advScreenNoHit{} advisers are assigned no by rule. This branch accounts for \advScreenNoHitSharePct{} per cent of all labels and \advScreenRuleNoSharePct{} per cent of no labels. After applying the ten author adjudications described below, the frozen primary counts are \advBOneTwoYes{} yes, \advBOneTwoNo{} no and \advBOneTwoUnclear{} unclear.

The validation reference is a non-random \advBOneTwoReference{}-adviser keyword-hit sample. The reference has \advBOneTwoOriginalAgreements{} labels from the earlier AI-agent reading on classifier agreements and \advBOneTwoAuthorAdjudicated{} author-adjudicated disagreements. Classifier precision is \advBOneTwoPrecision{} and recall is \advBOneTwoRecall{}. The comparison contains \advBOneTwoTP{} true positives, \advBOneTwoFP{} false positives, \advBOneTwoFN{} false negatives and \advBOneTwoTN{} true negatives. Appendix Table~\ref{tab:classifier-validation} reports exact binomial intervals, F1, specificity and accuracy with their denominators. Because only the \advBOneTwoAuthorAdjudicated{} original disagreements received independent author adjudication, these remain conditional validation statistics rather than accuracy against a fully author-labelled random sample. They also do not measure false negatives created by the screen. A separate \advScreenAuditN{}-adviser complete-text review used a fixed, reproducible ranking from the \advScreenAuditFrameN{}-adviser literal-screen no-hit group. An AI agent read each selected adviser's complete combined Part 2A text; the authors did not adjudicate those item-level judgments. Relative to that agent reference, the review found \advScreenAuditFalseNegatives{} agent-identified screen false negatives. The exact binomial (Clopper--Pearson) one-sided 95 per cent upper bound is \advScreenAuditUpperPct{} per cent relative to those agent judgments. It is not a bound against human ground truth, whole-pipeline recall or undisclosed real-world AI use. The complete-text review does not relabel the population.

\paragraph{Open-weights replication of the labels.} Under a plan time-stamped before the open-weights runs, we reapplied the frozen rule to the same \owNScreen{} snippets with two open-weights models, the primary Qwen3-235B (the A22B-Instruct-2507 release) and the smaller Qwen3-32B, at temperature zero. The replication used the same label rule and quote checks, although the frozen classifier processed advisers in batches and the open-weights models processed them one at a time. An answer that still failed after one retry or failed validation against the full brochure counted as unclear. For yes versus no or unclear, Cohen's kappa against the frozen classifier labels is \owKappaP{} (95 per cent interval \owKappaPLo{} to \owKappaPHi{}) for the primary model and \owKappaS{} for the smaller one, and the two open-weights models agree with each other at \owKappaPS{}. The plan specified that a primary kappa below 0.60 would be read as lack of robustness. Disagreements run in both directions: the primary model labels \owPYesRefNot{} advisers yes that the frozen classifier does not, and withholds yes from \owPNotRefYes{} that it labels yes. In an exploratory check, counting the models' labels whether or not their quotes validate raises the primary kappa only to \owKappaPLabelOnly{}. The disclosure measure is therefore not robust to the choice of classifier, and the comparisons below that use it describe the frozen classifier's labels. We retain the same primary labels as before the open-weights replication: the original classifier labels with the ten author verdicts added after the initial estimates. Appendix Table~\ref{tab:openweights} reports the agreement statistics and every label-dependent estimate under the primary open-weights labels.

\paragraph{Form 13F.} The holdings panel uses SEC Form 13F structured data from 2013Q2 through 2025Q2. The balanced sample contains \nBalFirst{} managers, with reported holdings of \$\aumBalFirst{} trillion in 2013Q2 and \$\aumBalLast{} trillion in 2025Q2. The final cleaning rule canonicalises numeric CIKs to the SEC's ten-character display form and keeps the latest original non-notice filing for each manager and holdings period. The holdings period must be a valid calendar quarter end; other dates are rejected before manager-period selection. The complete raw nine-character CUSIP must pass syntax and check-digit validation before the eight-character issue key is formed. Explicit filler patterns are excluded rather than treated as common assets. Repeated filing--CUSIP rows are aggregated. A missing, non-finite or non-positive component share count makes the aggregated share count unusable for implied-price evidence, while all positive reported values remain in portfolio weights. The parser skips malformed tab-separated records and replaces undecodable bytes but does not export a skipped-row count. We measure mean pairwise portfolio cosine, similarity to the value-weighted aggregate portfolio, and similarity of active positions after subtracting that aggregate, with and without the ten largest securities by aggregate value. Active-position cosines are defined only when every included portfolio differs numerically from the benchmark; a zero active-position vector makes the estimate undefined. A portfolio with no holding left after the top-ten exclusion is absent only from that exclusion measure.

The unit rule uses each report's median implied-price offset from period--CUSIP consensus prices. When no position has usable consensus evidence, the final implementation assigns the SEC date default before the report and within-manager jump screens. The within-manager screen uses only finite positive report totals. Reported values also apply filing-level unit corrections around the 2022Q4 reporting change and remove outliers. The main analysis retains original \texttt{13F-HR} filings and ignores \texttt{13F-HR/A} amendments. Complete restatements often differ materially from selected originals, so amendment-adjusted convergence estimates are not reported.

The statistical unit follows the source of identifying variation. N-PORT label comparisons retain fund or fund-pair observations and cluster at the adviser level, where the disclosure label is assigned. Dyadic regressions account for both adviser endpoints. Form 13F trend and break regressions treat quarterly aggregate outcomes as the time-series observations and use Newey--West/HAC standard errors with four lags. Intervals are two-sided 95 per cent intervals and $p$-values are two-sided, nominal and unadjusted across the multiple reported outcomes. The complete record of analysis plans, deviations and data corrections appears in Appendix~\ref{app:plan-reporting}.

\section{Empirical design}
\label{sec:design}

The N-PORT return comparison is an indirect descriptive diagnostic rather than an order- or signal-covariance test.

The empirical design uses two outcomes that address different parts of the stability question. For each eligible fund pair and month, let $W_t$ denote the twelve-month return window ending in month $t$. The return-comovement outcome is
\begin{equation}
\begin{aligned}
\widehat\rho_{ij,t}&=\operatorname{corr}\{r_{i\tau},r_{j\tau}\}_{\tau\in W_t},\\
\widetilde\rho_{ij,t}&=\min\{1-10^{-7},\max(-1+10^{-7},\widehat\rho_{ij,t})\},
&z_{ij,t}&=\operatorname{atanh}(\widetilde\rho_{ij,t}).
\end{aligned}
\label{eq:nport-outcome}
\end{equation}
where the pair enters only when the window contains twelve overlapping monthly returns. The primary contrast compares pairs whose advisers both disclose AI use with pairs whose advisers are both classified as having no detected disclosure. Mixed pairs and pairs advised by the same adviser are excluded.

The return-comovement regression conditions on common holdings and scale. We estimate
\begin{equation}
z_{ij,t}=\alpha_t+\beta_{AI}\,\mathrm{BothAI}_{ij}+\beta_H\,\mathrm{Cosine}_{ij,t}
+\beta_S\log\!\sqrt{A_{i,t}A_{j,t}}+\varepsilon_{ij,t},
\label{eq:nport-headline}
\end{equation}
where $\mathrm{Cosine}_{ij,t}$ is same-month equity-weight cosine, $A_{i,t}$ and $A_{j,t}$ are net assets, and $\alpha_t$ are calendar-month fixed effects. Adviser-level inference clusters dyadic score sums by both month-specific adviser endpoints and subtracts the unordered adviser-pair intersection \citep{AronowSamiiAssenova2015}. The primary interval uses a finite-adviser Student-$t$ distribution. The fund-endpoint sandwich and asymptotic-normal adviser interval are sensitivities.

The stress-window outcome measures overlap in selling rather than return comovement. For each event, we bracket the event with the last available holdings month before it and the first month after it within sixty-two calendar days. A fund-stock observation is a seller when its share change falls below the event-specific twentieth percentile of non-zero changes in the retained paired sample from the two comparison groups. Portfolio-weight change is a fallback when shares are unavailable and would use a separate percentile because the units are incomparable; all three reported windows use share changes, so this fallback is not used. Let $S_{sg}^{-i}$ be the AUM-weighted share of other funds in disclosure group $g$ that sell stock $s$. The focal fund's overlap is
\begin{equation}
O_{ig}=\frac{\sum_{s\in H_i^{\mathrm{pre}}}w_{is}^{\mathrm{pre}}S_{sg}^{-i}}{\sum_{s\in H_i^{\mathrm{pre}}}w_{is}^{\mathrm{pre}}},
\label{eq:sales-overlap}
\end{equation}
and the reported estimate is the mean difference in $O_{ig}$ between the two disclosure groups. The three windows are March 2020, August 2024 and April 2025. The adviser-cluster percentile bootstrap holds each event's realised seller threshold fixed while recomputing seller shares, focal-fund leave-one-out overlap and group means. The intervals are therefore conditional on the threshold estimated in the retained paired sample rather than intervals for a statistic that re-estimates the percentile in every draw.

The Form 13F outcomes measure portfolio similarity under explicit benchmarks. We report mean pairwise cosine, similarity to the value-weighted aggregate portfolio, and active-position similarity after subtracting that aggregate. Managers above the 2013--2014 median active share, the distance between a manager's portfolio and the aggregate benchmark, form the active half; the remainder form the index-like half. Each measure is reported with and without the ten largest securities by aggregate value. Quarterly trends use the 2013Q2--2025Q2 balanced panel and Newey--West standard errors with four lags. The analysis-plan record and every deviation or correction are consolidated in Appendix~\ref{app:plan-reporting}.

\section{Results}
\label{sec:evidence}

\subsection{N-PORT return comovement and stress-window selling}
\label{sec:nport}

The N-PORT return comparison does not show higher conditional comovement for the disclosure group. The primary sample contains \nPortComovementAiFunds{} funds from \nPortComovementAiAdvisers{} advisers with detected disclosure and \nPortComovementNoFunds{} funds from \nPortComovementNoAdvisers{} advisers without detected disclosure. Across \nPortComovementPairWindows{} pair-windows, the coefficient is \nPortComovementCoef{} Fisher-$z$ units with a finite-adviser 95 per cent interval from \mbox{\nPortComovementLo} to \nPortComovementHi{}. The interval includes zero. The fund-endpoint sensitivity is \mbox{\nPortComovementFundLo} to \nPortComovementFundHi{}, which lies below zero; because the disclosure label is assigned to advisers, adviser-level inference is primary, and the dependence unit changes precision rather than the point estimate.

\begin{table}[t]
\centering
\caption{Primary N-PORT disclosure-group comparisons}
\label{tab:nport-primary}
\small
\resizebox{\textwidth}{!}{\begin{tabular}{p{4.0cm}rrrr}
\toprule
Outcome & Estimate & 95\% interval & Disclosure funds & Comparison funds \\
\midrule
Return comovement (Fisher-$z$) & \nPortComovementCoef{} & [\mbox{\nPortComovementLo},\ \nPortComovementHi{}] & \nPortComovementAiFunds{} & \nPortComovementNoFunds{} \\
March 2020 sales overlap & \nPortStressCovidDiff{} & [\nPortStressCovidLo{},\ \nPortStressCovidHi{}] & \nPortStressCovidAiFunds{} & \nPortStressCovidNoFunds{} \\
August 2024 sales overlap & \nPortStressYenDiff{} & [\nPortStressYenLo{},\ \nPortStressYenHi{}] & \nPortStressYenAiFunds{} & \nPortStressYenNoFunds{} \\
April 2025 sales overlap & \nPortStressTariffDiff{} & [\nPortStressTariffLo{},\ \nPortStressTariffHi{}] & \nPortStressTariffAiFunds{} & \nPortStressTariffNoFunds{} \\
\bottomrule
\end{tabular}}
\begin{minipage}{0.94\textwidth}\footnotesize
\textit{Notes.} The first row is a same-label fund-pair regression with finite-adviser inference. The remaining rows are fund-level overlap differences with adviser-cluster percentile-bootstrap intervals conditional on the seller thresholds estimated in the retained sample; the thresholds are not re-estimated across draws. Disclosure funds are advised by advisers classified as disclosing AI use; comparison funds are advised by advisers with no detected disclosure. Mixed pairs are excluded from the return comparison. Adviser counts and pair-window counts appear in Appendix Table~\ref{tab:nport-planned-diagnostics}. Intervals are nominal and unadjusted across outcomes.
\end{minipage}
\end{table}

The stress-window comparisons point upward while remaining imprecise. The disclosure-minus-comparison differences are \nPortStressCovidDiff{} in March 2020, \nPortStressYenDiff{} in August 2024 and \nPortStressTariffDiff{} in April 2025. Every interval in Table~\ref{tab:nport-primary} includes zero. Each window meets the minimum disclosure-fund count. The outcome is a holdings-overlap measure, and the lack of explicit split factors limits its interpretation. Although the procedure allows separate thresholds for share and portfolio-weight changes, the reported samples contain only share changes.

\begin{figure}[t]
\centering
\includegraphics[width=\textwidth]{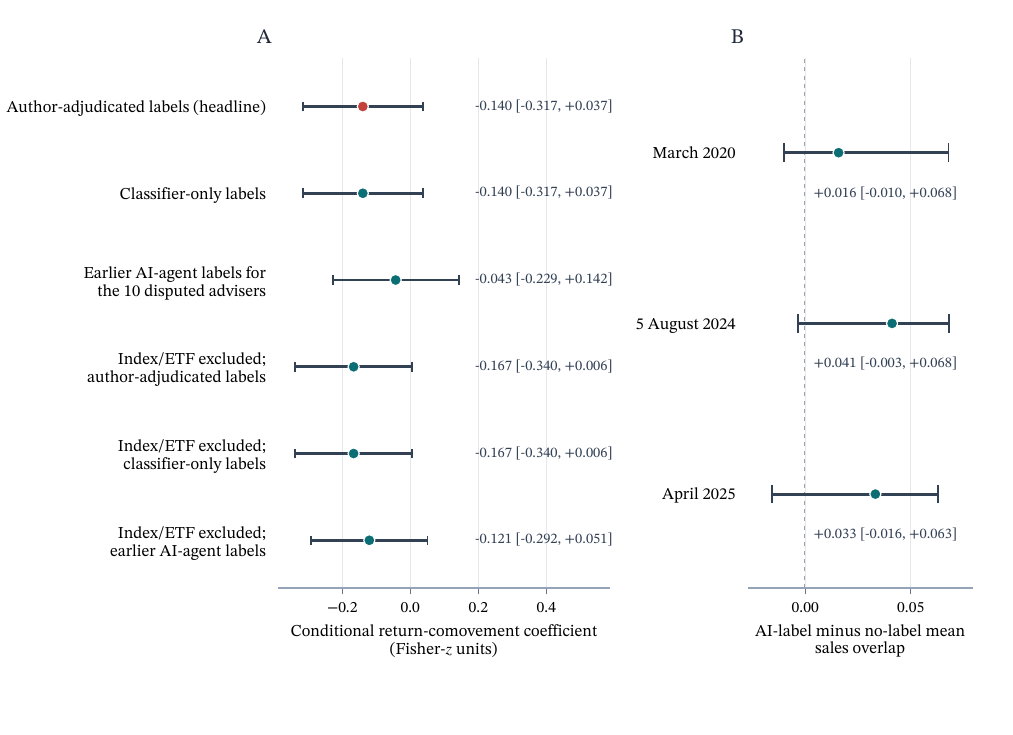}
\caption{\textbf{Conditional fund-return comovement and stress-window sales overlap.} Panel A reports the primary both-disclosure coefficient and label and index/ETF sensitivities. The comparison group contains pairs whose advisers are both classified as having no detected disclosure; mixed pairs are excluded. Lines are finite-adviser Student-$t$ intervals clustered by the two adviser endpoints. Panel B reports disclosure-minus-comparison fund-level sales-overlap differences with adviser-cluster percentile-bootstrap intervals that hold the realised event-specific seller threshold fixed. Counts and the exact outcome definitions appear in the text and Appendix Table~\ref{tab:nport-planned-diagnostics}. The comparisons are descriptive associations and do not observe common signals or orders. Source: public Form N-PORT/N-CEN and Form ADV.}
\label{fig:nport}
\end{figure}
\FloatBarrier

\subsection{Sensitivity to adviser disclosure classification}
\label{sec:robustness}

The return-comovement estimate is negative under both open-weights classifiers. Its interval excludes zero for Qwen3-235B (\owNPortCoefP{}; \owNPortLoP{} to \owNPortHiP{}), but not for Qwen3-32B (\owNPortCoefS{}; \owNPortLoS{} to \owNPortHiS{}). The Qwen3-235B interval also excludes zero after removing index funds and ETFs (\owBFiveExcludeCoefP{}; \owBFiveExcludeLoP{} to \owBFiveExcludeHiP{}). Appendix Table~\ref{tab:openweights} reports the full comparison.

Classifier-only labels give \nPortBThreeClassifierCoef{} with interval \nPortBThreeClassifierLo{} to \nPortBThreeClassifierHi{}. Applying the earlier AI-agent labels to the ten disputed advisers gives \nPortBThreeManualCoef{} with interval \nPortBThreeManualLo{} to \nPortBThreeManualHi{}. Agreement-only labels give \nPortBThreeAgreementCoef{} with interval \nPortBThreeAgreementLo{} to \nPortBThreeAgreementHi{}. Excluding index funds and ETFs gives \nPortBFiveExcludeCoef{} with interval \nPortBFiveExcludeLo{} to \nPortBFiveExcludeHi{}. The intervals for these four comparisons include zero. The expert-system boundary case has no exact-CRD N-PORT series and therefore has no targeted coefficient.

\subsection{Form 13F convergence and benchmark choice}
\label{sec:holdings}

The 13F trend is sensitive to the benchmark and the largest positions. In the balanced panel, manager-relative similarity rises \hMoneBalPct{} per cent and similarity to the aggregate portfolio rises \hMtwoBalPct{} per cent relative to the 2013--2014 base mean. Removing the ten largest holdings reduces the corresponding changes to \hTwoMoneExPctBal{} and \hTwoMtwoExPctBal{} per cent. Active-position similarity rises \xMoondBalAggPct{} per cent under the aggregate benchmark and \xMoondBalAvgPct{} per cent under the average-manager benchmark. The changing \$500 million sample reverses the sign, with \xMoondFiveAggPct{} per cent under the aggregate benchmark and \xMoondFiveAvgPct{} per cent under the average-manager benchmark. The balanced and changing samples therefore answer different questions.

The concentration and activity splits qualify the trend. Among the \nAct{} active-share managers, active-position similarity changes by \hMthreeActPct{} per cent with all stocks and by \hMthreexActPct{} per cent after the ten largest are removed. The index-like group changes by \hMthreexPasPct{} per cent after the same exclusion. The short extension through 2026Q2 contains positive ex-top-ten changes, including \eraActMthreexDiff{} percentage points for the active half and \eraBalMthreexDiff{} percentage points for the balanced panel. These later estimates use a short post period and changing coverage.

\begin{table}[t]
\centering
\caption{Active-position convergence under two benchmarks}
\label{tab:benchmarks}
\footnotesize
\setlength{\tabcolsep}{4pt}
\begin{tabular}{lllrr}
\toprule
Sample & Managers, first--last & Benchmark & Trend-fit (\%) & $p$ \\
\midrule
Balanced & 1,556--1,556 & Aggregate portfolio & +227.7 & $6.79\times10^{-6}$ \\
Balanced & 1,556--1,556 & Average manager & +27.6 & $1.75\times10^{-5}$ \\
AUM $\geq$ \$500M & 1,428--3,284 & Aggregate portfolio & +421.3 & $4.63\times10^{-10}$ \\
AUM $\geq$ \$500M & 1,428--3,284 & Average manager & -15.8 & $7.21\times10^{-4}$ \\
\bottomrule
\end{tabular}

\smallskip
\begin{minipage}{0.96\textwidth}\footnotesize
\textit{Notes.} The base is the 2013Q2--2014Q4 mean. Trend-fit change is the fitted linear-trend change through 2025Q2 divided by that base. Each regression uses quarterly aggregate observations with Newey--West standard errors and four lags. The balanced panel holds membership fixed; the \$500 million sample does not. The aggregate benchmark is value weighted and the average-manager benchmark is equal weighted.
\end{minipage}
\end{table}

\begin{table}[t]
\centering
\caption{Portfolio similarity, balanced panel and subsamples}
\label{tab:holdings}
\small
\setlength{\tabcolsep}{5pt}
\resizebox{\textwidth}{!}{\begin{tabular}{llrrrr}
\toprule
Sample & Measure & Base & 2025Q2 & Trend-fit (\%) & $p$ \\
\midrule
Balanced & $M_1$ manager-relative & 0.087 & 0.137 & +67.2 & $1.16\times10^{-55}$ \\
Balanced & $M_1$ ex top-10 & 0.066 & 0.071 & +12.7 & $8.31\times10^{-20}$ \\
Balanced & $M_2$ vs agg. & 0.286 & 0.357 & +30.5 & $1.09\times10^{-74}$ \\
Balanced & $M_2$ vs agg., ex top-10 & 0.247 & 0.258 & +7.7 & $3.90\times10^{-18}$ \\
Balanced & Active share & 0.834 & 0.782 & -6.7 & $8.67\times10^{-94}$ \\
Active half & $M_3$ active positions, ex top-10 & 0.013 & 0.013 & -5.4 & $0.113$ \\
Index-like half & $M_3$ active positions, ex top-10 & 0.022 & 0.014 & -40.0 & $5.21\times10^{-69}$ \\
\bottomrule
\end{tabular}
}
\smallskip
\begin{minipage}{0.95\textwidth}\footnotesize
\textit{Notes.} Trend-fit changes are relative to the 2013Q2--2014Q4 base mean. The balanced cohort contains \nBalFirst{} managers. The active and index-like halves are defined by the 2013--2014 median active share. Ex-top-ten measures remove the ten largest securities by aggregate value each quarter and renormalise the remaining portfolio.
\end{minipage}
\end{table}

\begin{figure}[t]
\centering
\includegraphics[width=\textwidth]{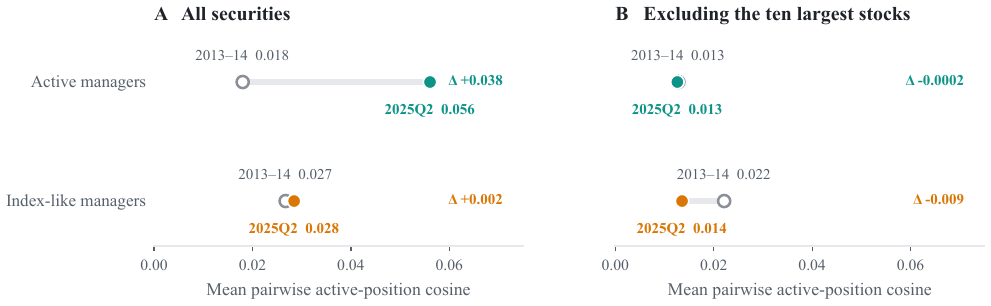}
\caption{\textbf{Active-position convergence depends on the largest stocks.} Each dumbbell joins the mean group-level active-position cosine in 2013Q2--2014Q4 to the fixed-cohort endpoint in 2025Q2. Panel A uses all securities. Panel B removes the ten largest securities by aggregate 13F value each quarter and renormalises portfolios. The figure displays endpoints without uncertainty intervals; quarterly trend inference appears in Tables~\ref{tab:holdings} and~\ref{tab:benchmarks}. Source: SEC Form 13F structured data sets.}
\label{fig:holdings}
\end{figure}
\FloatBarrier

From 2013Q2 through 2025Q2, the increase is concentrated in positions that include the ten largest stocks. The active-position trend changes when the benchmark or the sample changes.

The results remain sensitive to the benchmark after excluding filings without consensus-price evidence. Dropping reports without consensus-price evidence changes the balanced cohort from \unitNoConsensusPrimaryBalancedN{} to \unitNoConsensusDropBalancedN{} managers and changes the fitted balanced $M_1$ and $M_2$ trends from \unitNoConsensusMonePrimaryPct{} to \unitNoConsensusMoneDropPct{} per cent and from \unitNoConsensusMtwoPrimaryPct{} to \unitNoConsensusMtwoDropPct{} per cent. The adviser-level comparison specified for 13F is not estimated because labelled advisers cover only \hFourBOneTwoCoveragePct{} per cent of balanced-panel value, below the coverage criterion. The later descriptive comparisons and full sensitivity inventory appear in Appendix~Appendix~\ref{app:plan-reporting}.

The supplementary anomaly-loading and break analyses remain descriptive. The active-half ex-top-ten estimate is \hMthreexActPct{} per cent in the main window, while the constrained break condition is not met for that outcome. The extension contains only \eraNAgAct{} quarters, from 2025Q3 through 2026Q2. Post hoc free-break estimates are reported with their short post window and nominal inference. None identifies an effect of AI adoption.

\subsection{Limitations}
\label{sec:limitations}

The disclosure measure is current and incomplete. Current disclosures may be selective and do not date adoption. A no-detected-disclosure label does not imply that an adviser does not use AI, and the data do not systematically identify deployed model versions, shared training data or decision inputs. The validation reference combines \advBOneTwoOriginalAgreements{} earlier AI-agent labels with \advBOneTwoAuthorAdjudicated{} author-adjudicated disagreements. The separate \advScreenAuditN{}-adviser complete-text review also uses AI-agent judgments that the authors did not adjudicate. Its exact binomial (Clopper--Pearson) \advScreenAuditUpperPct{} per cent bound is relative to those judgments. It is conservative for the finite adviser frame and is not human-ground-truth validation. The screen also omits the term \emph{expert system}. Broad fund-category controls are unavailable in the N-PORT panel.

The trading outcomes are indirect. Reported holdings lack explicit split factors. The sales procedure permits a portfolio-weight fallback under a separate event-specific threshold, but all three reported windows use share changes and do not invoke that fallback. Their bootstrap intervals are conditional on the seller thresholds estimated in the retained analysis samples; uncertainty from estimating those thresholds is not included. The 13F panel uses original filings even though complete restatements often differ materially, and the convergence estimates have not been recomputed under an amendment-combination rule. The parser does not export skipped-row counts. Among the \unitNoConsensusPrimaryReports{} reports assigned the date-default unit, the first unit diagnostic is missing for \unitNoConsensusDOneMissing{}; both available cases point against the default, but that diagnostic can misclassify low-price securities. The 13F adviser comparison is descriptive because labelled advisers cover \hFourBOneTwoCoveragePct{} per cent of balanced-panel value. The extension through 2026Q2 and post hoc break specifications use short post periods and nominal, unadjusted inference.

The mechanism is uncalibrated and conditional. Its fold conditions apply above one-half adoption, and the illustrative shift is not an estimate of institutional behaviour. The computed bound holds only under a fixed adoption menu along the same fold path. These limits leave the paper with a measurement conclusion. Current public filings do not establish or exclude common-model fragility. The expert-system boundary case shows that the keyword screen misses an explicitly described use, but it cannot be estimated in N-PORT.

\section{Conclusion}
\label{sec:conclusion}

Public filings do not provide precise evidence of greater return comovement or common selling for funds linked to advisers with detected AI disclosure. The N-PORT point estimate is negative, while the three stress-window sales-overlap estimates are positive. Every interval under the primary labels includes zero. The negative return-comovement interval excludes zero with Qwen3-235B labels but not with Qwen3-32B labels. These comparisons do not observe common signals, orders or adoption timing.

The 13F evidence shows why portfolio convergence requires a benchmark and concentration diagnosis. The balanced-panel active-position increase falls from \xMoondBalAggPct{} to \xMoondBalAvgPct{} per cent when the benchmark changes from the aggregate portfolio to the average manager. The sign reverses in the changing \$500 million sample. Much of the primary-window movement involves the largest stocks, while the short extension contains positive ex-top-ten changes. The result is therefore sensitive to the sample and the object being measured.

The practical implication is a data requirement. Supervisors need deployment dates, model vendors and versions, data sources, the investment decisions that use model outputs, and records of human changes or overrides. Those fields would turn a cross-sectional disclosure comparison into a design that can compare institutions over time and identify common exposures. Until those records exist, public filings cannot attribute systemic fragility to shared AI models.

\section*{Funding}
This research did not receive any specific grant from funding agencies in the public, commercial, or not-for-profit sectors.

\section*{Declaration of competing interest}
 The authors declare that they have no known competing financial interests or personal relationships that could have appeared to influence the work reported in this paper.

\section*{Declaration of generative AI and AI-assisted technologies in the manuscript preparation process}
During the preparation of this work the authors used Claude Code (Anthropic; Claude Opus 5.5 and Claude Sonnet 5) and Codex (OpenAI; GPT-5.6-sol) in order to draft and edit text, write analysis and figure code, classify adviser brochures, and search for references. The brochure work included a separate AI-agent reading, the classifier, and a 50-adviser complete-text review by an AI agent. The author adjudicated the 10 cases on which the classifier and the earlier AI-agent reading disagreed; the other 40 reference labels are the earlier AI-agent labels on cases where the two methods agreed. The authors did not independently adjudicate the separate complete-text review judgments. After using these tools, the authors reviewed and edited the content as needed and take full responsibility for the content of the published article.

\section*{Data and code availability}
The study uses public SEC Form 13F, Form N-PORT and Form N-CEN structured data sets and Form ADV Part 2A brochures and bulk data \citep{SEC13FData2026,SECNPORTData2026,SECNCENData2026,SECIAPDFormADV}; the CBOE VIX series distributed by FRED \citep{FREDVIXCLS2026}; JKP U.S. monthly value-weighted long--short anomaly returns and the JKP U.S. market return from Global Factor Data \citep{JKPGlobalFactorData2026,JensenKellyPedersen2023}; and the Kenneth French monthly risk-free rate \citep{FrenchDataLibrary2026}. The arXiv source package contains the analysis code, a fixed VIX snapshot and derived analysis files needed to regenerate all reported tables and figures. Full re-estimation additionally requires downloading the public SEC and IAPD records and the JKP and Kenneth French return files from the cited sources; the raw SEC archives, adviser brochures and external return files are not redistributed, and the derived fund panels and frozen factor-input assembly must then be reconstructed following the replication instructions.

\bibliographystyle{plainnat}
\bibliography{references/bibliography}

\appendix
\section{Analysis plans, deviations and corrections}
\label{app:plan-reporting}

Table~\ref{tab:analysis-plan-record} reports when each analysis was specified, its final design and the interpretation supported by that design. It distinguishes analyses fixed before estimation from later exploratory work and retains results that were not estimated because their stated criteria were not met.

Under the superseded keyword rule, an earlier separate AI-agent reading labelled \fpNo{} of the \fpN{} keyword-positive brochures as not stating the adviser's own investment use, an agent-relative rate of \fpRate{} per cent. Those labels were not author-adjudicated and do not estimate the current classifier's false-positive rate.

\scriptsize
\setlength{\tabcolsep}{2pt}
\begin{longtable}{@{}p{2.0cm}p{2.0cm}p{3.4cm}p{2.8cm}@{}}
\caption{Analysis plans, deviations and corrections}\label{tab:analysis-plan-record}\\
\toprule
Analysis & Analysis status & Final specification & Interpretation \\
\midrule
\endfirsthead
\toprule
Analysis & Analysis status & Final specification & Interpretation \\
\midrule
\endhead
\midrule\multicolumn{4}{r}{Continued on next page}\\
\endfoot
\bottomrule
\endlastfoot
N-PORT return comparison & Specified before the primary coefficient, after panel construction & The public panel lacks a broad fund-category control. & The primary coefficient is conditional on holdings overlap, fund size and month, but not broad fund category. \\
N-PORT inference & Specified after the initial results & Inference uses month-specific adviser endpoints, dyadic adviser clustering and a finite-adviser Student-$t$ reference distribution; fund-endpoint inference is a sensitivity. & The point estimate is unchanged; the adviser-level interval includes zero, while the fund-endpoint interval lies below zero. \\
Stress-window sales overlap & Specified before the three event-window estimates & Seller status uses event-specific share-change thresholds; a separate portfolio-weight threshold is available as a fallback but is not used in the reported windows. Adviser-cluster draws hold the realised thresholds fixed. & The positive differences are imprecise, and their intervals are conditional on the retained-sample thresholds. Missing split factors remain a measurement limitation. \\
Initial adviser disclosure labels & Before linking the initial labels to N-PORT & The screen assigns most no-detected-disclosure labels by rule; the initial classifier and AI-agent readings disagreed in ten cases. & The initial estimates use the pre-adjudication labels. \\
Author adjudication and final label rerun & After the initial N-PORT results but before the final rerun & The author adjudicated the ten disagreements; the other forty reference labels retain their earlier AI-agent provenance. & Primary labels use the adjudications. Classifier-only, earlier-agent and agreement-only estimates remain sensitivities. \\
Expert-system boundary & Post-registration and before checking N-PORT coverage & Only an explicitly self-described expert system used in investment research or decisions qualifies. & The boundary case has \nPortBSixExpertTargetSeries{} exact-CRD N-PORT series, so no coefficient is estimated. \\
N-PORT identifier cleaning & Before the final panel was re-estimated & Invalid and filler CUSIPs are excluded before aggregation. Optional identifiers remain audit fields and are not fallback keys. & Dependent N-PORT estimates use the cleaned panel. \\
13F identifier and unit cleaning & Before dependent 13F results were recomputed & The final path rejects malformed identifiers, retains positive weights when share counts are unusable, and assigns the SEC date default when no consensus-price evidence exists. & The balanced cohort and trend estimates use this final path. A drop-no-consensus sensitivity is descriptive. \\
Original versus amended 13F filings & Original plan before estimation & The primary analysis keeps original filings. A comparison shows material differences in complete restatements, but no amendment-combination rule was estimated. & No amendment-adjusted coefficient is reported. The originals-only policy is a limitation. \\
Adviser-level 13F comparison & Specified before estimation & The labelled adviser group covers \hFourBOneTwoCoveragePct{} per cent of balanced-panel value, below the required coverage criterion. & The planned estimate is not run. Below-threshold comparisons are labelled descriptive and have no registered verdict. \\
Benchmark, concentration and era extensions & Specified before their respective estimates, after earlier 13F results were known & The balanced and changing samples use different membership rules, and the extension through 2026Q2 has a short post period. & These results qualify the primary-window concentration account and do not identify AI adoption effects. \\
Free-break sensitivity & Specified after the main estimates were known and exploratory & The sensitivity separates the level jump from the slope change and reports nominal inference across alternative lag choices. & It does not replace the registered constrained break or identify an AI effect. \\
\end{longtable}

The revised 13F specification was fixed before its reported estimates but after earlier 13F evidence was known, so it is not independent of that earlier evidence. The initial N-PORT specification was fixed after panel construction but before the primary coefficient. Streaming changes the computation rather than the estimand and was adopted when the initial implementation proved infeasible.

The official Form N-CEN bulk files omit filings submitted under one filing schema, which limits adviser-link coverage.

The final identifier sample excludes \nCusipExcludedWindow{} source rows and identifies \nCusipPlaceholderCodesWindow{} explicit placeholder codes. It retains \nDthreeMissingAll{} reports without a finite total; \nDthreeMissingScope{} fall in the screen window and affect \nDthreeMissingManagers{} manager histories. A further \nDthreeInsufficient{} positive-total reports have fewer than four usable comparison reports and are counted rather than screened. The final sample ignores every \texttt{13F-HR/A} amendment, canonicalises numeric CIKs to the SEC's ten-character display form, keeps the latest original non-notice filing, validates the complete raw nine-character CUSIP, and aggregates repeated filing--CUSIP rows. It rejects unassigned \texttt{000000}/\texttt{999999} issuer prefixes rather than treating unknown securities as common assets. A missing, non-finite or non-positive component share count makes the aggregated share count unusable for implied-price evidence, while all positive reported values remain in portfolio weights. The parser skips malformed tab-separated records. A portfolio with no holding left after the top-ten exclusion is absent only from that exclusion measure. No amendment-adjusted coefficient is reported. The no-consensus sensitivity uses \unitNoConsensusPrimaryReports{}, \unitNoConsensusPostSwitchReports{}, \unitNoConsensusSfiveReports{} and \unitNoConsensusBalancedReports{} reports and compares cohorts of \unitNoConsensusPrimaryBalancedN{} and \unitNoConsensusDropBalancedN{} managers. Among the \unitNoConsensusPrimaryReports{} reports assigned the date default, the first unit diagnostic is missing for \unitNoConsensusDOneMissing{} and points against the default for all \unitNoConsensusDOneAvailable{} reports where it is available. Because that diagnostic can misclassify low-price securities, this contrary evidence does not establish the correct alternative unit.

Validation uses the non-random \advBOneTwoReference{}-adviser keyword-hit reference. The separate no-hit review reports a binomial (Clopper--Pearson) bound that is conservative for the finite adviser frame. The agreement-only interval is narrower than the headline interval, while the full-population interval using the earlier AI-agent labels is slightly wider; neither excludes zero.

\begin{table}[htbp]
\centering
\caption{Classifier validation against the frozen 50-adviser reference set}
\label{tab:classifier-validation}
\small
\begin{tabular}{@{}llll@{}}
\toprule
Metric & Estimate & Exact 95\% interval & Count basis \\
\midrule
Precision & \advBOneTwoPrecision{} & [\advBOneTwoPrecisionLo{}, \advBOneTwoPrecisionHi{}] & \advBOneTwoTP{}/(\advBOneTwoTP{}+\advBOneTwoFP{}) \\
Recall & \advBOneTwoRecall{} & [\advBOneTwoRecallLo{}, \advBOneTwoRecallHi{}] & \advBOneTwoTP{}/(\advBOneTwoTP{}+\advBOneTwoFN{}) \\
F1 & \advBOneTwoFOne{} & --- & $2\mathrm{TP}/(2\mathrm{TP}+\mathrm{FP}+\mathrm{FN})$ \\
Specificity & \advBOneTwoSpecificity{} & --- & \advBOneTwoTN{}/(\advBOneTwoTN{}+\advBOneTwoFP{}) \\
Accuracy & \advBOneTwoAccuracy{} & --- & (\advBOneTwoTP{}+\advBOneTwoTN{})/\advBOneTwoReference{} \\
\bottomrule
\end{tabular}

\smallskip

\begin{minipage}{0.94\textwidth}\footnotesize
\textit{Notes.} The independent unit is adviser; $n=\advBOneTwoReference{}$ advisers form the non-random keyword-hit reference set. Forty labels come from the earlier AI-agent reading on cases where it agreed with the classifier; the author adjudicated the 10 original disagreements. Precision and recall intervals are two-sided 95 per cent Clopper--Pearson intervals. F1, specificity and accuracy are point statistics. The two intervals are nominal and not multiplicity-adjusted. Because only the original disagreements received independent author adjudication, these quantities are not accuracy against a fully author-labelled random sample, full-population screening recall or evidence about undisclosed real-world AI use.
\end{minipage}
\end{table}

\begin{table}[htbp]
\centering
\caption{Planned power and sample diagnostics for N-PORT sensitivities}
\label{tab:nport-planned-diagnostics}
\scriptsize
\resizebox{\textwidth}{!}{\begin{tabular}{@{}p{3.0cm}rrrrrrrr@{}}
\toprule
Analysis & Estimate & SE & 95\% interval & 80\% MDE & Funds Y/N & Advisers Y/N & Pair-windows YY/all & Floor \\
\midrule
March 2020 sales overlap & \nPortStressCovidDiff{} & \nPortStressCovidBootstrapSE{} & [\nPortStressCovidLo{}, \nPortStressCovidHi{}] & \nPortStressCovidMde{} & \nPortStressCovidAiFunds{}/\nPortStressCovidNoFunds{} & \nPortStressCovidAiAdvisers{}/\nPortStressCovidNoAdvisers{} & --- & \nPortStressCovidPowerFloor{} \\
August 2024 sales overlap & \nPortStressYenDiff{} & \nPortStressYenBootstrapSE{} & [\nPortStressYenLo{}, \nPortStressYenHi{}] & \nPortStressYenMde{} & \nPortStressYenAiFunds{}/\nPortStressYenNoFunds{} & \nPortStressYenAiAdvisers{}/\nPortStressYenNoAdvisers{} & --- & \nPortStressYenPowerFloor{} \\
April 2025 sales overlap & \nPortStressTariffDiff{} & \nPortStressTariffBootstrapSE{} & [\nPortStressTariffLo{}, \nPortStressTariffHi{}] & \nPortStressTariffMde{} & \nPortStressTariffAiFunds{}/\nPortStressTariffNoFunds{} & \nPortStressTariffAiAdvisers{}/\nPortStressTariffNoAdvisers{} & --- & \nPortStressTariffPowerFloor{} \\
\midrule
Agreement-only advisers & \nPortBThreeAgreementCoef{} & \nPortBThreeAgreementSE{} & [\nPortBThreeAgreementLo{}, \nPortBThreeAgreementHi{}] & \nPortBThreeAgreementMde{} & \nPortBThreeAgreementAiFunds{}/\nPortBThreeAgreementNoFunds{} & \nPortBThreeAgreementAiAdvisers{}/\nPortBThreeAgreementNoAdvisers{} & \nPortBThreeAgreementAiPairWindows{}/\nPortBThreeAgreementPairWindows{} & \nPortBThreeAgreementPowerFloor{} \\
Classifier-only labels & \nPortBThreeClassifierCoef{} & \nPortBThreeClassifierSE{} & [\nPortBThreeClassifierLo{}, \nPortBThreeClassifierHi{}] & \nPortBThreeClassifierMde{} & \nPortBThreeClassifierAiFunds{}/\nPortBThreeClassifierNoFunds{} & \nPortBThreeClassifierAiAdvisers{}/\nPortBThreeClassifierNoAdvisers{} & \nPortBThreeClassifierAiPairWindows{}/\nPortBThreeClassifierPairWindows{} & \nPortBThreeClassifierPowerFloor{} \\
Earlier AI-agent labels for 10 disputed advisers & \nPortBThreeManualCoef{} & \nPortBThreeManualSE{} & [\nPortBThreeManualLo{}, \nPortBThreeManualHi{}] & \nPortBThreeManualMde{} & \nPortBThreeManualAiFunds{}/\nPortBThreeManualNoFunds{} & \nPortBThreeManualAiAdvisers{}/\nPortBThreeManualNoAdvisers{} & \nPortBThreeManualAiPairWindows{}/\nPortBThreeManualPairWindows{} & \nPortBThreeManualPowerFloor{} \\
\midrule
Exclude index/ETF, author-adjudicated labels & \nPortBFiveExcludeCoef{} & \nPortBFiveExcludeSE{} & [\nPortBFiveExcludeLo{}, \nPortBFiveExcludeHi{}] & \nPortBFiveExcludeMde{} & \nPortBFiveExcludeAiFunds{}/\nPortBFiveExcludeNoFunds{} & \nPortBFiveExcludeAiAdvisers{}/\nPortBFiveExcludeNoAdvisers{} & \nPortBFiveExcludeAiPairWindows{}/\nPortBFiveExcludePairWindows{} & \nPortBFiveExcludePowerFloor{} \\
Full sample with index-pair controls & \nPortBFiveControlsCoef{} & \nPortBFiveControlsSE{} & [\nPortBFiveControlsLo{}, \nPortBFiveControlsHi{}] & \nPortBFiveControlsMde{} & \nPortBFiveControlsAiFunds{}/\nPortBFiveControlsNoFunds{} & \nPortBFiveControlsAiAdvisers{}/\nPortBFiveControlsNoAdvisers{} & \nPortBFiveControlsAiPairWindows{}/\nPortBFiveControlsPairWindows{} & \nPortBFiveControlsPowerFloor{} \\
Exclude index/ETF, classifier-only labels & \nPortBFiveExcludeClassifierCoef{} & \nPortBFiveExcludeClassifierSE{} & [\nPortBFiveExcludeClassifierLo{}, \nPortBFiveExcludeClassifierHi{}] & \nPortBFiveExcludeClassifierMde{} & \nPortBFiveExcludeClassifierAiFunds{}/\nPortBFiveExcludeClassifierNoFunds{} & \nPortBFiveExcludeClassifierAiAdvisers{}/\nPortBFiveExcludeClassifierNoAdvisers{} & \nPortBFiveExcludeClassifierAiPairWindows{}/\nPortBFiveExcludeClassifierPairWindows{} & \nPortBFiveExcludeClassifierPowerFloor{} \\
Exclude index/ETF, earlier AI-agent labels & \nPortBFiveExcludeManualCoef{} & \nPortBFiveExcludeManualSE{} & [\nPortBFiveExcludeManualLo{}, \nPortBFiveExcludeManualHi{}] & \nPortBFiveExcludeManualMde{} & \nPortBFiveExcludeManualAiFunds{}/\nPortBFiveExcludeManualNoFunds{} & \nPortBFiveExcludeManualAiAdvisers{}/\nPortBFiveExcludeManualNoAdvisers{} & \nPortBFiveExcludeManualAiPairWindows{}/\nPortBFiveExcludeManualPairWindows{} & \nPortBFiveExcludeManualPowerFloor{} \\
\bottomrule
\end{tabular}
}
\smallskip

\begin{minipage}{0.98\textwidth}\footnotesize
\textit{Notes.} Y/N gives the fund counts and the adviser-level clustering counts; YY/all gives yes--yes and all same-label fund-pair windows. Sales-overlap estimates, standard errors, intervals and MDEs are in fund-level overlap units; their 95 per cent intervals use adviser-cluster percentile bootstraps with \nPortStressCovidBootstrapDraws{} valid draws per window, while their MDEs use a normal approximation based on the bootstrap standard error. The other seven rows are Fisher-$z$ return-comovement coefficients with two-sided finite-adviser Student-$t$ intervals from a dyadic sandwich clustered by the two adviser endpoints, and two-sided noncentral-$t$ MDEs. ``Floor'' reports whether the analysis-plan minimum sample size was met. Author-adjudicated labels are primary; classifier-only and earlier-AI-agent rows are sensitivities. Intervals are nominal and unadjusted across rows, and no sensitivity receives a new directional verdict.
\end{minipage}
\end{table}

\begin{table}[htbp]
\centering
\caption{Below-threshold descriptive Form 13F disclosure comparisons}
\label{tab:h4-diagnostics}
\scriptsize
\resizebox{\textwidth}{!}{\begin{tabular}{@{}lrrrrrr@{}}
\toprule
Descriptive comparison & Groups Y/N & Analysed AUM (\%) & Change Y & Change N & Difference & 95\% interval \\
\midrule
Yes versus no & \hFourBOneTwoPrimaryAi{}/\hFourBOneTwoPrimaryNo{} & \hFourBOneTwoPrimaryAumPct{} & \hFourBOneTwoPrimaryChAi{} & \hFourBOneTwoPrimaryChNon{} & \hFourBOneTwoPrimaryDiff{} & [\hFourBOneTwoPrimaryLo{}, \hFourBOneTwoPrimaryHi{}] \\
Yes plus unclear versus no & \hFourBOneTwoRobustAi{}/\hFourBOneTwoRobustNo{} & \hFourBOneTwoRobustAumPct{} & \hFourBOneTwoRobustChAi{} & \hFourBOneTwoRobustChNon{} & \hFourBOneTwoRobustDiff{} & [\hFourBOneTwoRobustLo{}, \hFourBOneTwoRobustHi{}] \\
\bottomrule
\end{tabular}
}
\smallskip

\begin{minipage}{0.98\textwidth}\footnotesize
\textit{Notes.} Group counts are balanced-panel managers linked to adviser-level labels. Changes compare mean within-group active-position cosine across \hFourBOneTwoPrimaryPreQuarters{} pre-period and \hFourBOneTwoPrimaryPostQuarters{} post-period quarters. The quarter is the independent inferential unit. The 95 per cent intervals use \hFourBOneTwoPrimaryBootstrapDraws{} joint non-circular moving-period-block draws of length \hFourBOneTwoPrimaryBlockLength{}, with identical sampled quarters for both groups. The registered comparison was infeasible because brochure-covered advisers account for only \hFourBOneTwoCoveragePct{} per cent of balanced-panel 2025Q2 value, below the prespecified 30 per cent coverage criterion. These two later comparisons are descriptive only; their intervals are nominal and unadjusted across the two rows, and neither receives the registered hypothesis verdict.
\end{minipage}
\end{table}

\begin{table}[htbp]
\centering
\caption{Open-weights replication of the disclosure labels}
\label{tab:openweights}
\scriptsize
\resizebox{\textwidth}{!}{\begin{tabular}{@{}lrrrrr@{}}
\toprule
\multicolumn{6}{@{}l}{\textit{Panel A. Agreement with the frozen classifier labels (\owNScreen{} snippet advisers)}} \\
Label source & Binary $\kappa$ & 95\% interval & Agreement (\%) & Invalid answers & Yes/no/unclear \\
\midrule
Frozen classifier & -- & -- & -- & -- & \owRefYes{}/\owRefNo{}/\owRefUnclear{} \\
Qwen3-235B (primary) & \owKappaP{} & [\owKappaPLo{}, \owKappaPHi{}] & \owAgreePPct{} & \owInvalidP{} & \owPYes{}/\owPNo{}/\owPUnclear{} \\
Qwen3-32B & \owKappaS{} & [\owKappaSLo{}, \owKappaSHi{}] & \owAgreeSPct{} & \owInvalidS{} & \owSYes{}/\owSNo{}/\owSUnclear{} \\
Qwen3-235B vs Qwen3-32B & \owKappaPS{} & [\owKappaPSLo{}, \owKappaPSHi{}] & -- & -- & -- \\
\midrule
\multicolumn{6}{@{}l}{\textit{Panel B. Label-dependent estimates}} \\
Estimate & \multicolumn{1}{r}{Frozen labels} & 95\% interval & \multicolumn{1}{r}{Qwen3-235B} & 95\% interval & \\
\midrule
N-PORT comovement, all funds & \nPortComovementCoef{} & [\nPortComovementLo{}, \nPortComovementHi{}] & \owNPortCoefP{} & [\owNPortLoP{}, \owNPortHiP{}] & \\
N-PORT comovement, non-index/ETF & \nPortBFiveExcludeCoef{} & [\nPortBFiveExcludeLo{}, \nPortBFiveExcludeHi{}] & \owBFiveExcludeCoefP{} & [\owBFiveExcludeLoP{}, \owBFiveExcludeHiP{}] & \\
Sales overlap, March 2020$^{a}$ & \nPortStressCovidDiff{} & [\nPortStressCovidLo{}, \nPortStressCovidHi{}] & \owStressCovidDiffP{} & [\owStressCovidLoP{}, \owStressCovidHiP{}] & \\
Sales overlap, August 2024$^{a}$ & \nPortStressYenDiff{} & [\nPortStressYenLo{}, \nPortStressYenHi{}] & \owStressYenDiffP{} & [\owStressYenLoP{}, \owStressYenHiP{}] & \\
Sales overlap, April 2025$^{a}$ & \nPortStressTariffDiff{} & [\nPortStressTariffLo{}, \nPortStressTariffHi{}] & \owStressTariffDiffP{} & [\owStressTariffLoP{}, \owStressTariffHiP{}] & \\
13F cosine change gap, yes vs no & \hFourBOneTwoPrimaryDiff{} & [\hFourBOneTwoPrimaryLo{}, \hFourBOneTwoPrimaryHi{}] & \owHFourPrimaryDiffP{} & [\owHFourPrimaryLoP{}, \owHFourPrimaryHiP{}] & \\
13F cosine change gap, yes+unclear vs no & \hFourBOneTwoRobustDiff{} & [\hFourBOneTwoRobustLo{}, \hFourBOneTwoRobustHi{}] & \owHFourRobustDiffP{} & [\owHFourRobustLoP{}, \owHFourRobustHiP{}] & \\
\midrule
\multicolumn{6}{@{}l}{\textit{Panel C. 13F diagnostics without author verdicts}} \\
Estimate & Frozen classifier & 95\% interval & Qwen3-235B & 95\% interval & \\
13F cosine change gap, yes vs no & \owHFourPrimaryDiffRef{} & [\owHFourPrimaryLoRef{}, \owHFourPrimaryHiRef{}] & \owHFourPrimaryDiffPOnly{} & [\owHFourPrimaryLoPOnly{}, \owHFourPrimaryHiPOnly{}] & \\
13F cosine change gap, yes+unclear vs no & \owHFourRobustDiffRef{} & [\owHFourRobustLoRef{}, \owHFourRobustHiRef{}] & \owHFourRobustDiffPOnly{} & [\owHFourRobustLoPOnly{}, \owHFourRobustHiPOnly{}] & \\
\bottomrule
\end{tabular}
}
\smallskip

\begin{minipage}{0.98\textwidth}\footnotesize
\textit{Notes.} Panel A: the independent unit is adviser; $n=\owNScreen{}$ advisers with a literal-screen hit, whose snippets all classifiers read under the same frozen rule. Binary $\kappa$ is Cohen's kappa for yes versus no or unclear against the frozen classifier labels, with a 95 per cent percentile interval from 2,000 adviser bootstrap resamples. Invalid answers are those whose quoted sentence still failed validation after one retry; they count as unclear, as in the frozen rule. Counting each model's labels regardless of quote validity gives $\kappa=\owKappaPLabelOnly{}$ (\owKappaPLabelOnlyLo{} to \owKappaPLabelOnlyHi{}) for Qwen3-235B and \owKappaSLabelOnly{} (\owKappaSLabelOnlyLo{} to \owKappaSLabelOnlyHi{}) for Qwen3-32B. Panel B: frozen-label estimates are the estimates using the frozen labels with the ten author verdicts; the open-weights estimates retain the ten author verdicts and replace the other model labels, leaving rule-based no labels and all non-label inputs fixed. For the sales-overlap rows and the panel below, the open-weights labels do not incorporate the ten author verdicts. The two N-PORT comovement estimates are identical with and without those verdicts under each model. Panel C compares the 13F diagnostics without any author verdicts on either side; all 13F rows compare groups on the change in within-group active-position cosine from the pre-period to the post-period defined in Appendix Table~\ref{tab:h4-diagnostics}. They are descriptive and below the coverage threshold for the registered test. Qwen3-32B labels give \owNPortCoefS{} (\owNPortLoS{} to \owNPortHiS{}) and \owBFiveExcludeCoefS{} (\owBFiveExcludeLoS{} to \owBFiveExcludeHiS{}) for the two N-PORT rows and \owHFourPrimaryDiffS{} (\owHFourPrimaryLoS{} to \owHFourPrimaryHiS{}) and \owHFourRobustDiffS{} (\owHFourRobustLoS{} to \owHFourRobustHiS{}) for the two 13F rows. $^{a}$The sales-overlap rows were produced by the same estimation run but were not among the label-dependent estimates listed in the plan; they are reported in addition.
\end{minipage}
\end{table}

\begin{table}[htbp]
\centering
\caption{Fund-level anomaly-loading changes by current disclosure label}
\label{tab:a2-anomaly-loading}
\scriptsize
\textit{Panel A: Change in anomaly loading across eras}\par\smallskip
\resizebox{\textwidth}{!}{%
\begin{tabular}{@{}p{2.75cm}rrrrrrrr@{}}
\toprule
Sample & $\gamma$ & SE & 95\% interval & 80\% MDE & Funds Y/N & Advisers Y/N & Floor & Reading \\
\midrule
Non-index/ETF, yes versus no & \nPortATwoPrimaryGamma{} & \nPortATwoPrimarySE{} & [\nPortATwoPrimaryLo{}, \nPortATwoPrimaryHi{}] & \nPortATwoPrimaryMde{} & \nPortATwoPrimaryAiFunds{}/\nPortATwoPrimaryNoFunds{} & \nPortATwoPrimaryAiAdvisers{}/\nPortATwoPrimaryNoAdvisers{} & \nPortATwoPrimaryPowerFloor{} & \nPortATwoPrimaryVerdict{} \\
Non-index/ETF, yes plus unclear versus no & \nPortATwoUnclearGamma{} & \nPortATwoUnclearSE{} & [\nPortATwoUnclearLo{}, \nPortATwoUnclearHi{}] & \nPortATwoUnclearMde{} & \nPortATwoUnclearAiFunds{}/\nPortATwoUnclearNoFunds{} & \nPortATwoUnclearAiAdvisers{}/\nPortATwoUnclearNoAdvisers{} & \nPortATwoUnclearPowerFloor{} & \nPortATwoUnclearVerdict{} \\
Index/ETF, yes versus no & \nPortATwoIndexGamma{} & \nPortATwoIndexSE{} & [\nPortATwoIndexLo{}, \nPortATwoIndexHi{}] & \nPortATwoIndexMde{} & \nPortATwoIndexAiFunds{}/\nPortATwoIndexNoFunds{} & \nPortATwoIndexAiAdvisers{}/\nPortATwoIndexNoAdvisers{} & \nPortATwoIndexPowerFloor{} & \nPortATwoIndexVerdict{} \\
\bottomrule
\end{tabular}}\par\medskip
\textit{Panel B: Descriptive yes-label minus comparison-group level differences}\par\smallskip
\resizebox{\textwidth}{!}{%
\begin{tabular}{@{}p{5.0cm}rr@{}}
\toprule
Sample & Pre era: estimate [95\% interval] & Post era: estimate [95\% interval] \\
\midrule
Non-index/ETF, yes versus no & \nPortATwoPrimaryPreDiff{} [\nPortATwoPrimaryPreLo{}, \nPortATwoPrimaryPreHi{}] & \nPortATwoPrimaryPostDiff{} [\nPortATwoPrimaryPostLo{}, \nPortATwoPrimaryPostHi{}] \\
Non-index/ETF, yes plus unclear versus no & \nPortATwoUnclearPreDiff{} [\nPortATwoUnclearPreLo{}, \nPortATwoUnclearPreHi{}] & \nPortATwoUnclearPostDiff{} [\nPortATwoUnclearPostLo{}, \nPortATwoUnclearPostHi{}] \\
Index/ETF, yes versus no & \nPortATwoIndexPreDiff{} [\nPortATwoIndexPreLo{}, \nPortATwoIndexPreHi{}] & \nPortATwoIndexPostDiff{} [\nPortATwoIndexPostLo{}, \nPortATwoIndexPostHi{}] \\
\bottomrule
\end{tabular}}

\smallskip

\begin{minipage}{0.98\textwidth}\footnotesize
\textit{Notes.} For each fund and era, excess return is regressed on the market and an equal-weighted composite of published U.S. anomaly returns. $\gamma$ is the yes-label fund change from October 2019--October 2022 to November 2022--December 2025 relative to the comparison group. Author-adjudicated labels are assigned at adviser level, which is also the independent clustering unit; the table reports both fund and adviser $n$. Standard errors are adviser-clustered CR1, and intervals are two-sided asymptotic-normal 95 per cent intervals. The primary non-index/ETF estimate is \nPortATwoPrimaryGamma{} (\nPortATwoPrimaryLo{} to \nPortATwoPrimaryHi{}), so the registered verdict is uninformative; its minimum sample floor is met. The other rows and pre/post level differences are descriptive; intervals are nominal and unadjusted across rows. The index/ETF row has only \nPortATwoIndexAiFunds{} yes-label funds from \nPortATwoIndexAiAdvisers{} advisers and does not meet the primary floor. Labels come from current 2025--26 brochures, so a pre-era yes label means that the adviser later stated AI use. Style drift, flows and a fixed latest-adviser assignment can also change measured loadings. The comparison does not identify adoption timing or an AI effect.
\end{minipage}
\end{table}

\clearpage
\subsection{Complete sensitivity inventory}
\label{app:complete-sensitivities}

The following tables report the planned, descriptive and post hoc results. Rows marked as not estimated remain distinct from zero estimates, and post hoc rows retain their nominal inference status.

\begingroup\scriptsize\setlength{\tabcolsep}{3pt}
\begin{longtable}{p{1.4cm}p{5.8cm}p{5.5cm}}
\caption{Registered results and post hoc sensitivities for the Form 13F analyses}\label{tab:registered-results}\\
\toprule
Sample & Result & Outcome or limitation \\
\midrule
\endfirsthead
\toprule
Sample & Result & Outcome or limitation \\
\midrule
\endhead
\midrule\multicolumn{3}{r}{Continued on next page}\\\endfoot
\bottomrule\endlastfoot
Balanced & \textit{Full window}. $M_1$ fitted change (\%); +67.2; $p=1.16\times10^{-55}$ & supported \\
Balanced & \textit{Excluding 2020}. $M_1$ fitted change (\%); +66.4; $p=3.30\times10^{-48}$ & supported \\
Balanced & \textit{Full window}. $M_2$ fitted change (\%); +30.5; $p=1.09\times10^{-74}$ & supported \\
Balanced & \textit{Excluding 2020}. $M_2$ fitted change (\%); +30.2; $p=2.39\times10^{-67}$ & supported \\
Balanced & \textit{Full window}. $M_3$ fitted change (\%); +227.7; $p=6.79\times10^{-6}$ & supported \\
Balanced & \textit{Excluding 2020}. $M_3$ fitted change (\%); +228.0; $p=8.78\times10^{-6}$ & supported \\
AUM $\geq$ \$500M & \textit{Full window}. $M_1$ fitted change (\%); +44.9; $p=1.14\times10^{-20}$ & supported \\
AUM $\geq$ \$500M & \textit{Excluding 2020}. $M_1$ fitted change (\%); +44.0; $p=2.53\times10^{-16}$ & supported \\
AUM $\geq$ \$500M & \textit{Full window}. $M_2$ fitted change (\%); +17.2; $p=9.46\times10^{-22}$ & supported \\
AUM $\geq$ \$500M & \textit{Excluding 2020}. $M_2$ fitted change (\%); +16.7; $p=7.33\times10^{-17}$ & supported \\
AUM $\geq$ \$500M & \textit{Full window}. $M_3$ fitted change (\%); +421.3; $p=4.63\times10^{-10}$ & supported \\
AUM $\geq$ \$500M & \textit{Excluding 2020}. $M_3$ fitted change (\%); +425.6; $p=6.81\times10^{-11}$ & supported \\
All filers & \textit{Full window}. $M_1$ fitted change (\%); +44.1; $p=1.31\times10^{-15}$ & supported \\
All filers & \textit{Excluding 2020}. $M_1$ fitted change (\%); +42.9; $p=1.31\times10^{-13}$ & supported \\
All filers & \textit{Full window}. $M_2$ fitted change (\%); +14.9; $p=1.76\times10^{-11}$ & supported \\
All filers & \textit{Excluding 2020}. $M_2$ fitted change (\%); +14.4; $p=3.54\times10^{-10}$ & supported \\
All filers & \textit{Full window}. $M_3$ fitted change (\%); +365.2; $p=4.97\times10^{-11}$ & supported \\
All filers & \textit{Excluding 2020}. $M_3$ fitted change (\%); +366.1; $p=6.13\times10^{-11}$ & supported \\
Balanced & \textit{Top-ten exclusion}. $M_2$ full to ex-top-10; +30.5 to +7.7\%; $p=3.90\times10^{-18}$ & concentration \\
All filers & \textit{Top-ten exclusion}. $M_2$ full to ex-top-10; +14.9 to -9.0\%; $p=2.59\times10^{-6}$ & concentration \\
Balanced & \textit{Full window}. VIX coefficient $\times10^3$; 0.283; $p=0.203$ & not supported \\
Balanced & \textit{Excluding 2020Q2}. VIX coefficient $\times10^3$; 0.196; $p=0.435$ & not supported \\
AUM $\geq$ \$500M & \textit{Full window}. VIX coefficient $\times10^3$; 0.282; $p=0.461$ & not supported \\
AUM $\geq$ \$500M & \textit{Excluding 2020Q2}. VIX coefficient $\times10^3$; 0.101; $p=0.803$ & not supported \\
Balanced & \textit{Earlier keyword labels}. Earlier keyword-based disclosure split; not estimated; coverage 27.5\%; earlier keyword-based AI-investment labels 305 & infeasible (criteria not met) \\
Balanced & \textit{Current disclosure comparison}. Author-adjudicated disclosure split; yes 151, no 802, unclear excluded 33; coverage 27.5\% & not estimated; minimum coverage criterion not met \\
Active half & \textit{Active-manager trend}. $M_3$ fitted change (\%); +211.3; $p=5.00\times10^{-12}$ & not distinguishable from passive growth and concentration \\
Active half & \textit{Active-manager trend}. $M_3$ ex top-10 fitted change (\%); -5.4; $p=0.113$ & not distinguishable from passive growth and concentration \\
Active half & \textit{Constrained break}. $M_3$ registered constrained break $c\times100$; 0.025; $p=0.005$ & registered positive-break condition met; level jump constrained to $39c$; 49 quarters, 10 post \\
Active half & \textit{\eraAiOnlyWindow{} slope}. $M_3$ post-break slope $\times100$; 0.198; $p=9.99\times10^{-13}$ & prespecified reporting only; outside the decision rule \\
Active half & \textit{Post hoc level}. $M_3$ free level jump $\times100$; 0.444; $p=0.147$ & post hoc; independent level and slope; HAC(4) \\
Active half & \textit{Post hoc slope}. $M_3$ free slope change $\times100$; 0.141; $p=2.10\times10^{-6}$ & post hoc; HAC lag: nominal $p$ 0: $7.62\times10^{-4}$, 1: $5.44\times10^{-5}$, 2: $1.37\times10^{-5}$, 4: $2.10\times10^{-6}$; 10 post quarters; no causal identification \\
Active half & \textit{Constrained break}. $M_3$ ex top-10 registered constrained break $c\times100$; 0.001; $p=0.218$ & registered positive-break condition not met; level jump constrained to $39c$; 49 quarters, 10 post \\
Active half & \textit{\eraAiOnlyWindow{} slope}. $M_3$ ex top-10 post-break slope $\times100$; 0.010; $p=0.048$ & prespecified reporting only; outside the decision rule \\
Active half & \textit{Post hoc level}. $M_3$ ex top-10 free level jump $\times100$; 0.000; $p=1.000$ & post hoc; independent level and slope; HAC(4) \\
Active half & \textit{Post hoc slope}. $M_3$ ex top-10 free slope change $\times100$; 0.013; $p=0.023$ & post hoc; HAC lag: nominal $p$ 0: $0.014$, 1: $0.027$, 2: $0.028$, 4: $0.023$; 10 post quarters; no causal identification \\
Active half & \textit{Joint condition}. Additional joint condition; no acceleration & extra conservative rule, not a jointly specified registered condition: both outcomes positive with $p<0.05$ \\
Balanced & \textit{Constrained break}. $M_3$ registered constrained break $c\times100$; 0.019; $p=2.52\times10^{-4}$ & registered positive-break condition met; level jump constrained to $39c$; 49 quarters, 10 post \\
Balanced & \textit{\eraAiOnlyWindow{} slope}. $M_3$ post-break slope $\times100$; 0.174; $p=7.05\times10^{-20}$ & prespecified reporting only; outside the decision rule \\
Balanced & \textit{Post hoc level}. $M_3$ free level jump $\times100$; 0.138; $p=0.239$ & post hoc; independent level and slope; HAC(4) \\
Balanced & \textit{Post hoc slope}. $M_3$ free slope change $\times100$; 0.155; $p=1.72\times10^{-15}$ & post hoc; HAC lag: nominal $p$ 0: $3.22\times10^{-7}$, 1: $5.05\times10^{-11}$, 2: $6.87\times10^{-13}$, 4: $1.72\times10^{-15}$; 10 post quarters; no causal identification \\
Balanced & \textit{Constrained break}. $M_3$ ex top-10 registered constrained break $c\times100$; 0.001; $p=0.327$ & registered positive-break condition not met; level jump constrained to $39c$; 49 quarters, 10 post \\
Balanced & \textit{\eraAiOnlyWindow{} slope}. $M_3$ ex top-10 post-break slope $\times100$; 0.011; $p=0.063$ & prespecified reporting only; outside the decision rule \\
Balanced & \textit{Post hoc level}. $M_3$ ex top-10 free level jump $\times100$; -0.029; $p=0.485$ & post hoc; independent level and slope; HAC(4) \\
Balanced & \textit{Post hoc slope}. $M_3$ ex top-10 free slope change $\times100$; 0.016; $p=0.009$ & post hoc; HAC lag: nominal $p$ 0: $0.037$, 1: $0.047$, 2: $0.027$, 4: $0.009$; 10 post quarters; no causal identification \\
Balanced & \textit{Joint condition}. Additional joint condition; no acceleration & extra conservative rule, not a jointly specified registered condition: both outcomes positive with $p<0.05$ \\
Balanced & \textit{Observed change}. $M_3$ fitted change and 95\% interval; 0.017781 [0.010037, 0.025525] & unmatched adoption proxies; bound changes only AI-labelled-manager pairs \\
Balanced & \textit{Illustrative score}. Moond (2026), 2019 brochure vintage; earlier squared-dose score ($\Delta_{AA}=0.36$); $\phi$=0.008500; score=0.000026 & fixed calibration below observed interval \\
Balanced & \textit{Count bound}. Moond (2026), 2019 brochure vintage; conditional count PSD/OLS bound; $\phi$=0.008500; bound=0.000119 & conditional count bound below interval; changing only AI-labelled-manager pairs is insufficient \\
Balanced & \textit{Illustrative score}. Moond (2026), 2024 brochure vintage; earlier squared-dose score ($\Delta_{AA}=0.36$); $\phi$=0.041000; score=0.000605 & fixed calibration below observed interval \\
Balanced & \textit{Count bound}. Moond (2026), 2024 brochure vintage; conditional count PSD/OLS bound; $\phi$=0.041000; bound=0.002487 & conditional count bound below interval; changing only AI-labelled-manager pairs is insufficient \\
Balanced & \textit{Illustrative score}. AIMA (2025), broad work-use scenario; earlier squared-dose score ($\Delta_{AA}=0.36$); $\phi$=0.950000; score=0.324900 & fixed calibration above observed interval \\
Balanced & \textit{Count bound}. AIMA (2025), broad work-use scenario; conditional count PSD/OLS bound; $\phi$=0.950000; bound=1.328412 & conditional count bound does not exclude the observed interval \\
Balanced & \textit{Illustrative score}. earlier keyword-based AI-investment count; earlier squared-dose score ($\Delta_{AA}=0.36$); $\phi$=0.196015; score=0.013832 & fixed calibration inside observed interval \\
Balanced & \textit{Count bound}. earlier keyword-based AI-investment count; conditional count PSD/OLS bound; $\phi$=0.196015; bound=0.056493 & conditional count bound does not exclude the observed interval \\
Balanced & \textit{Illustrative score}. earlier keyword-based AI-investment AUM; earlier squared-dose score ($\Delta_{AA}=0.36$); $\phi$=0.169956; score=0.010399 & fixed calibration inside observed interval \\
Balanced & \textit{Count bound}. earlier keyword-based AI-investment AUM; conditional count PSD/OLS bound; $\phi$=0.169956; bound=--- & not evaluated: AUM share does not identify equal-weight manager pairs \\
Balanced & \textit{Score threshold}. Legacy-score lower-CI crossing dose; 0.166974 & $\sqrt{L/0.36}$; arithmetic score crossing only, without a channel or OLS mapping \\
Balanced & \textit{Count-bound threshold}. Conditional count-bound lower-CI crossing dose; 0.082905 (129 managers) & $\min_n\{K n^2/[N(N-1)]\geq L\}/N$; conditional reachability, not power or rejection \\
Active half & \textit{Era extension}. $M_1$ coefficients $\times100$ (percentage points); $c_1$=0.074, $c_2$=0.318, diff.=0.245; $p=0.001$ & MDEs=0.146/0.223/0.248 pp; 4 later quarters; managers 778 to 729 \\
Active half & \textit{Era extension}. $M_2$ coefficients $\times100$ (percentage points); $c_1$=-0.280, $c_2$=-0.121, diff.=0.160; $p=0.434$ & MDEs=0.408/0.626/0.695 pp; 4 later quarters; managers 778 to 729 \\
Active half & \textit{Era extension}. $M_3$ coefficients $\times100$ (percentage points); $c_1$=0.848, $c_2$=1.589, diff.=0.741; $p=0.012$ & MDEs=0.432/0.663/0.737 pp; 4 later quarters; managers 778 to 729 \\
Active half & \textit{Era extension}. $M_3$ ex top-10 coefficients $\times100$ (percentage points); $c_1$=0.065, $c_2$=0.242, diff.=0.177; $p=0.023$ & MDEs=0.059/0.090/0.100 pp; 4 later quarters; managers 778 to 729 \\
Balanced & \textit{Era extension}. $M_1$ coefficients $\times100$ (percentage points); $c_1$=0.128, $c_2$=0.025, diff.=-0.103; $p=0.702$ & MDEs=0.502/0.771/0.857 pp; 4 later quarters; managers 1,556 to 1,469 \\
Balanced & \textit{Era extension}. $M_2$ coefficients $\times100$ (percentage points); $c_1$=-0.138, $c_2$=-0.768, diff.=-0.630; $p=0.117$ & MDEs=0.721/1.107/1.230 pp; 4 later quarters; managers 1,556 to 1,469 \\
Balanced & \textit{Era extension}. $M_3$ coefficients $\times100$ (percentage points); $c_1$=0.633, $c_2$=1.457, diff.=0.824; $p=0.001$ & MDEs=0.147/0.226/0.251 pp; 4 later quarters; managers 1,556 to 1,469 \\
Balanced & \textit{Era extension}. $M_3$ ex top-10 coefficients $\times100$ (percentage points); $c_1$=0.026, $c_2$=0.225, diff.=0.198; $p=0.030$ & MDEs=0.051/0.079/0.088 pp; 4 later quarters; managers 1,556 to 1,469 \\
\end{longtable}
\begin{minipage}{0.98\textwidth}\footnotesize\textit{Notes.} The trend and break regressions use quarterly aggregate outcomes: 49 quarters from 2013Q2 through 2025Q2 unless a row states another window; era-extension rows use the 53 quarters from 2013Q2 through 2026Q2. The fixed balanced cohort contains 1,556 managers; the active-half cohort contains 778 managers at the registered-window start. The quarter is the time-series observation; manager counts describe the cross-section underlying each quarterly aggregate. Newey--West/HAC standard errors use 4 lags unless a row names another lag. Reported $p$-values are two-sided, nominal and unadjusted for the multiple reported outcomes; prespecified verdicts follow the analysis-plan decision rules. The 2019 and 2024 Moond values are Form ADV Part 2A brochure-mention shares among 586 advisers observed in both vintages; they measure disclosure, not adoption \citep{Moond2026}. The AIMA value is broad GenAI work use among 150 fund-manager respondents, not an investment-adoption share \citep{AIMA2025}. Neither external population is matched to the balanced 13F-manager panel. The fixed 0.36 score retains the earlier squared-dose arithmetic without a pair-weight or trend mapping. The post-registration count bound applies to a fixed cohort and allows only correlations between two AI-labelled managers to change: $q_{AA}=n_A(n_A-1)/[N(N-1)]$, PSD implies $\bar M_{AA}\geq-1/(n_A-1)$, and $B=K n_A^2/[N(N-1)]$ for $n_A\geq2$. The actual OLS fitted-change weights give $K=1.469388$. Active vectors are nonzero and normalized; correlations for all other pairs are held constant. External shares map hypothetically to $n_A=\lceil\phi N\rceil$; AUM share has no such count bound. This is a conservative outer bound, not a sharp portfolio maximum; a high bound does not establish sufficiency. Era-extension rows are descriptive and their $p$-values are nominal. Their coefficients and MDEs are multiplied by 100 and read as percentage points. The middle, later and difference minimum detectable differences refer to \eraMiddleWindow{}, \eraLaterWindow{} and their difference. They use the pre-era residual standard deviation, the plan's iid 2.8 multiplier and the number of quarters in each period (the later window has 4). They are approximate diagnostics: serial correlation and the fitted-trend design are not incorporated.\end{minipage}
\endgroup

\begin{table}[htbp]
\centering
\caption{Sensitivity to reports without consensus-price unit evidence}
\label{tab:no-consensus}
\small
\resizebox{\textwidth}{!}{\begin{tabular}{lrrrr}
\toprule
Measure & Date-default primary & $p$ & Drop no-consensus & $p$ \\
\midrule
Balanced $M_1$ & +67.2 & $1.16\times10^{-55}$ & +67.0 & $1.18\times10^{-54}$ \\
Balanced $M_2$ & +30.5 & $1.09\times10^{-74}$ & +30.5 & $3.75\times10^{-73}$ \\
Active-half $M_3$ & +211.3 & $5.00\times10^{-12}$ & +213.6 & $3.00\times10^{-12}$ \\
Active-half $M_3$ ex top-10 & -5.4 & $0.113$ & -5.2 & $0.129$ \\
\bottomrule
\end{tabular}
}
\smallskip
\begin{minipage}{0.94\textwidth}\footnotesize
\textit{Notes.} Entries are fitted linear-trend changes through 2025Q2 relative to the 2013Q2--2014Q4 base mean. The sensitivity drops every report without a consensus-price statistic before recomputing the within-manager screen and downstream samples. The $p$-values are nominal and unadjusted across rows.
\end{minipage}
\end{table}

\clearpage
\section{Mechanism and proofs}
\label{sec:mechanism}

The mechanism is a maintained price-projection schedule that clarifies how common signals and adoption incentives could interact. A share $\phi$ of $N$ institutions uses $x_i=v+\rho\eta+\sqrt{1-\rho^2}\nu_i$, while the remaining institutions observe $h_i=v+\epsilon_i$, and total order flow is $Y=\sum_iq_i+u$. The primitive shocks $v$, $\eta$, $u$, $\{\nu_i\}$ and $\{\epsilon_i\}$ are mean zero and pairwise orthogonal, including across institutions' idiosyncratic errors, with variances $\sigma_v^2$, $\sigma_\eta^2$, $\sigma_u^2$, $\sigma_\nu^2$ and $s$. Under fixed aggregate aggressiveness in a Kyle-style pricing block \citep{Kyle1985},
\begin{equation}
\lambda(\phi)=\frac{\sigma_v^3}{\sigma_u(2\sigma_v^2+\phi^2\rho^2\sigma_\eta^2)},
\qquad r(\phi)=\frac{\phi\rho\beta}{\lambda(\phi)+\phi\rho\beta/N},
\qquad \mathcal M=(1-r)^{-1}.
\label{eq:lambda}
\end{equation}
Let $\lambda'(\phi)\equiv\lambda(\phi)+\phi\rho\beta/N$ denote the feedback-adjusted schedule; the prime is a label, not a derivative. This schedule is maintained rather than derived from a full multi-trader equilibrium. The feedback domain requires $\lambda'>0$ and $|r|<1$. When $\rho\beta>0$, feedback rises with adoption. The adoption block uses a cdf $G$ of heterogeneous net costs and a quadratic career cost of departing from peer reliance, motivated by \citet{Scharfstein1990}. The fold conditions, stability restrictions and the two-logistic illustration are not calibrated to the filing data.

The cognitive-dependency state follows $\ell_{t+1}=(1-\delta_f)\ell_t+\delta_f d_i(t)$ and $s_t=\{\bar\tau(1-\ell_t)\}^{-1}$, where $\bar\tau=1/s_0$ is the fixed initial precision. Under the stated fixed-exposure schedule and a continued fold path above one-half adoption,
\begin{equation}
\frac{dc_j}{ds}=-\frac{\gamma\chi_H}{2\Delta^2(2\phi_j-1)}<0,
\qquad j\in\{E,X\}.
\label{eq:foldshift}
\end{equation}
The numerical illustration moves the exit fold by \propAddedExitSharePct{} per cent of the intact-skill loop. Evaluating both folds at the degraded state gives width \propTerminalLoopWidth{}, which is \propSameStateNarrowPct{} per cent narrower than the intact-skill width \propStaticLoopWidth{}. This same-state comparison differs from the historical comparison between entry with intact skill and exit with degraded skill, whose additional gap is positive. For any finite horizon and depreciation rate with the illustrated fixed adoption menu, $0\le d_i(t)\le\Delta=\propParamDelta{}$, the fixed-menu limit of \propAddedExitShareBoundPct{} per cent is an upper bound while the same two-fold branch and sign conditions continue. Lowering common signal correlation can raise adoption incentives within the maintained mean-payoff block, so the mechanism does not establish a correlation-cap policy result. These values use a dimensionless, uncalibrated normalisation.

The illustration closes the cost distribution as
\begin{equation}
G(x)=w\Lambda\!\left(\frac{x-\mu_L}{\sigma_c}\right)
+(1-w)\Lambda\!\left(\frac{x-\mu_H}{\sigma_c}\right),
\qquad \Lambda(z)=(1+e^{-z})^{-1}.
\label{eq:illustrative-cdf}
\end{equation}

\begin{table}[t]
\centering
\caption{Complete uncalibrated closure for the mechanism illustration}
\label{tab:illustrative-parameters}
\small
\begin{tabular}{p{3.1cm}p{2.4cm}p{3.1cm}p{2.8cm}}
\toprule
Parameter & Value & Parameter & Value \\
\midrule
$\sigma_v$ & \propParamSigmaV{} & $\sigma_u$ & \propParamSigmaU{} \\
$\sigma_\eta$ & \propParamSigmaEta{} & $\rho$ & \propParamRho{} \\
$\beta$ & \propParamBeta{} & $N$ & \propParamN{} \\
$\tau$ & \propParamTau{} & $\gamma$ & \propParamGamma{} \\
$\chi_H$ & \propParamChiH{} & $V_{H0}$ & \propParamVarHZero{} \\
$V_A$ & \propParamVarAi{} & $\Delta$ & \propParamDelta{} \\
$s_0$ & \propInitialVariance{} & $w$ & \propMixWeight{} \\
$\mu_L$ & \propMixLocationLow{} & $\mu_H$ & \propMixLocationHigh{} \\
$\sigma_c$ & \propMixScale{} & $\delta_f$ & \propDepreciation{} \\
$T$ & \propHorizon{} & $d_i(t)$ & \propParamDelta{} \\
\bottomrule
\end{tabular}
\end{table}

\paragraph{Price impact.} Institutions submit price-independent linear orders $q_i=b_sx_i$. With $A=\phi Nb_{AI}+(1-\phi)Nb_H$, total order flow has
\[
\Var(Y)=A^2\sigma_v^2+(\phi Nb_{AI})^2\rho^2\sigma_\eta^2
+\phi Nb_{AI}^2(1-\rho^2)\sigma_\nu^2
+(1-\phi)Nb_H^2s+\sigma_u^2,
\]
and the competitive linear projection is $\lambda=\Cov(v,Y)/\Var(Y)$. Under the common-aggressiveness normalisation used in the model, the thick-market limit gives \eqref{eq:lambda}. This is a maintained schedule rather than a solved Nash equilibrium.

\paragraph{Adoption folds.} Let $\tau>0$ be the payoff scale and let the human-profit variance be $V_H(s)=V_{H0}+\chi_Hs$ while adopter variance is fixed in $s$. The adoption incentive is
\[
\Delta U(\phi,c,s)=m_A-m_H-\frac{\gamma}{2}\{V_A(\phi)-V_H(s)\}+c\Delta^2(2\phi-1).
\]
The marginal switcher satisfies $c_i=\Delta U(\phi,c,s)=G^{-1}(\phi)$ and an equilibrium solves $\Phi(\phi,c,s)=G(\Delta U)-\phi=0$. We assume two non-degenerate folds $0<\phi_m<\phi_M<1$, with $\Phi_\phi$ negative, then positive, then negative across the two critical points, and the alternating function-value signs that give one crossing in each intervening interval. Under discrete best-response adjustment, the two outer crossings additionally satisfy $|G'(\Delta U)\Delta U_\phi|<1$, while the middle crossing satisfies $G'(\Delta U)\Delta U_\phi>1$. These conditions generate the stated static loop; they are maintained assumptions rather than a general existence result. The career term changes adoption in the required direction only above $\phi=1/2$, so an entry fold there also requires $G\{\Delta U_0(1/2,s)\}>1/2$ and removes an already AI-majority path.

\paragraph{Cognitive-dependency fold shift.} Starting from the full-use steady state gives the exact recursion
\[
\ell_T=\sum_{j=0}^{T-1}\delta_f(1-\delta_f)^{T-1-j}d_i(j).
\]
For bounded non-decreasing reliance paths with at least one positive reliance entry, $s_T>s_0$ and the implicit-function calculation yields \eqref{eq:foldshift}; the all-zero path instead gives $s_T=s_0$.

\clearpage
\section{Proofs}
\label{app:proofs}

\paragraph{Price impact.} Institutions submit price-independent linear orders $q_i=b_sx_i$. With $A=\phi Nb_{AI}+(1-\phi)Nb_H$, total order flow has
\[
\Var(Y)=A^2\sigma_v^2+(\phi Nb_{AI})^2\rho^2\sigma_\eta^2
+\phi Nb_{AI}^2(1-\rho^2)\sigma_\nu^2
+(1-\phi)Nb_H^2s+\sigma_u^2,
\]
and the competitive linear projection is
\[
\lambda=\frac{\Cov(v,Y)}{\Var(Y)}=\frac{A\sigma_v^2}{\Var(Y)}.
\]
Impose the common-aggressiveness normalisation $b_{AI}=b_H=b$ and $Nb=\sigma_u/\sigma_v$. The two idiosyncratic-signal terms are then $O(1/N)$, while
\[
\Var(Y)=\sigma_u^2\left(2+\frac{\phi^2\rho^2\sigma_\eta^2}{\sigma_v^2}\right)+O(1/N),
\]
which gives \eqref{eq:lambda} in the thick-market limit. This is a competitive linear-pricing relation at fixed aggregate aggressiveness, not a claim that $b$ solves a full multi-trader Nash problem. We likewise take $\lambda'=\lambda+\phi\rho\beta/N$ as the maintained feedback-adjusted schedule. Summing feedback rounds gives $\mathcal M=(1-r)^{-1}$ when $\lambda'>0$ and $|r|<1$.

\begin{lemma}[Amplification]
\label{lem:coupling}
For every $N$, $r$ is strictly increasing in $\phi$ when $\rho\beta>0$. In the limit $N\to\infty$, if $\rho\beta\sigma_\eta>0$, $r$ is strictly convex on $(0,1]\cap\{|r|<1\}$. For finite $N$ it is strictly convex on every $[\phi_{\min},1]\cap\{|r|<1\}$ for sufficiently large $N$. If $\beta=0$, then $r\equiv0$.
\end{lemma}

\paragraph{Proof.} Write $b_f=\rho\beta$, $L=\lambda+b_f\phi/N$ and $r=b_f\phi/L$. Since
\[
r'(\phi)=\frac{b_f\{L-\phi L_\phi\}}{L^2}
=\frac{b_f\{\lambda-\phi\lambda_\phi\}}{L^2}>0,
\]
$r$ is strictly increasing when $b_f>0$. Direct differentiation again gives the general expression
\[
r''(\phi)=\frac{2b_f}{L^2}(-L_\phi)
+\frac{b_f\phi}{L^2}(-L_{\phi\phi})
+\frac{2b_f\phi}{L^3}L_\phi^2.
\]
Let $\lambda_0=\sigma_v/(2\sigma_u)$ and $k=\rho^2\sigma_\eta^2/\sigma_v^2$. From \eqref{eq:lambda}, $\lambda=\lambda_0(1+k\phi^2/2)^{-1}$. Hence, as $N\to\infty$,
\[
r(\phi)=\frac{\rho\beta}{\lambda_0}\left(\phi+\frac{k}{2}\phi^3\right),
\qquad
r''(\phi)=\frac{3\rho\beta k}{\lambda_0}\phi>0
\]
when $\rho\beta k>0$ and $\phi>0$. For finite $N$, $r''$ converges uniformly to this positive limit on every $[\phi_{\min},1]$ with $\phi_{\min}>0$, so strict convexity holds there for sufficiently large $N$. It cannot hold uniformly at zero: $r''(0)=-2\rho^2\beta^2/(N\lambda_0^2)<0$. Finally, $\beta=0$ gives $r\equiv0$. \qed

\paragraph{Adoption folds.} Let $\tau>0$ be the payoff-scale parameter in the maintained expected-profit schedules
\[
m_A=\frac{\sigma_v^2-\rho^2\phi^2\sigma_\eta^2(1-r)^{-2}}{2\tau\lambda'},
\qquad
m_H(\phi,s)=\frac{\sigma_v^2}{2\tau\lambda(\phi)}.
\]
Thus $m_H$ depends on $\phi$ but, under this maintained fixed-exposure schedule, $m_{H,s}=0$. In the $N\to\infty$ limit and holding $\lambda$ fixed, the feedback contribution is identified by differentiating with respect to $\beta$:
\[
\left.\frac{\partial m_A}{\partial\beta}\right|_{\lambda}
=\frac{\lambda\phi^3\rho^3\sigma_\eta^2}
{\tau(\beta\phi\rho-\lambda)^3}<0.
\]
For $\rho,\sigma_\eta,\phi>0$, $\beta\ge0$, $\lambda'>0$ and $|r|<1$,
\[
2\tau(m_A-m_H)=\sigma_v^2\!\left(\frac{1}{\lambda'}-\frac{1}{\lambda}\right)
-\frac{\rho^2\phi^2\sigma_\eta^2}{\lambda'(1-r)^2}<0.
\]
The mean-profit block therefore does not establish adoption complementarity.

To close the adoption model, impose mean--variance preferences with risk-aversion coefficient $\gamma>0$. At the institution's fixed risk-budget scale, let $V_A(\phi)=\Var(\pi_A)$ be independent of $s$ and specify the human-profit variance as
\[
V_H(s)=V_{H0}+\chi_Hs,\qquad \chi_H>0.
\]
Its economic content is that a noisier fallback forecast raises payoff variance at the mandated exposure; the result would require a different condition if institutions changed that exposure with $s$. The adoption incentive is
\[
\begin{aligned}
\Delta U(\phi,c,s)
&=m_A-m_H-\frac{\gamma}{2}\{V_A(\phi)-V_H(s)\}
+c\Delta^2(2\phi-1),\\
\Delta U_s&=\frac{\gamma\chi_H}{2}>0.
\end{aligned}
\]
The net payoff difference is $\Delta U-c_i$, so the marginal switcher satisfies $c_i=\Delta U=G^{-1}(\phi)$ and an equilibrium satisfies $\Phi(\phi,c,s)=G(\Delta U(\phi,c,s))-\phi=0$.

Write $\Delta U=\Delta U_0+c\Delta^2(2\phi-1)$. If $(\Delta U_0)_\phi$ is bounded below by $-K$, then $c>K/(2\Delta^2)$ makes the best response increasing and yields increasing differences and a greatest equilibrium \citep{Milgrom1990}. Under the derivative and value sign patterns assumed in the main text, the intermediate value theorem gives one root in each of $(0,\phi_m)$, $(\phi_m,\phi_M)$ and $(\phi_M,1)$, and the derivative pattern gives no others. We separately impose the local discrete-adjustment condition $|G'(\Delta U)\Delta U_\phi|<1$ at the outer crossings and $G'(\Delta U)\Delta U_\phi>1$ at the middle crossing. The outer equilibria are then stable under best-response adjustment and the middle one is unstable. At the entry fold, $\Phi_{\phi\phi}>0$; at the exit fold, $\Phi_{\phi\phi}<0$. At either fold
\[
\Phi_c=g(\Delta U)\Delta^2(2\phi-1)>0
\]
under the stated condition $\phi>1/2$. Raising $c$ therefore lifts the local minimum through zero at $c_E$, eliminating the lower two roots, while lowering $c$ pushes the local maximum through zero at $c_X$, eliminating the upper two. This proves the stated static loop. \qed

At $\phi=1/2$ the career term is zero for every $c$. If $\Phi(1/2,c,s)<0$ while $\Phi(0,c,s)>0$, a root below one half survives for every $c$ and cannot be removed at an entry fold above one half. Hence such an entry fold requires $G(\Delta U_0(1/2,s))>1/2$, and the entire lower equilibrium path that it removes is AI-majority.

\paragraph{Cognitive-dependency fold shift.} Write $\delta=\delta_f$, $\bar\tau=n_0\tau_s/\delta$ and $\ell_t=1-\tau_H(t)/\bar\tau$. Throughout this argument $0<\delta<1$, $0\le d_i(t)\le1$, $\ell_0=0$ and $\ell_t<1$. The final restriction is the domain on which $s_t=\{\bar\tau(1-\ell_t)\}^{-1}$ is finite and positive. Starting from the full-use steady state gives the exact precision recursion
\[
\ell_{t+1}=(1-\delta)\ell_t+\delta d_i(t),\qquad
\ell_{t+1}-\ell_t=\delta\{d_i(t)-\ell_t\}.
\]
Iteration first gives
\[
\ell_T=\sum_{j=0}^{T-1}\delta(1-\delta)^{T-1-j}d_i(j).
\]
Therefore any reliance path in this bounded domain with one positive entry has $\ell_T>0$ and $s_T>s_0$; monotonicity of the path is not needed for this sign. For the depreciation comparative static, assume the path is non-decreasing. Abel summation gives
\[
\begin{aligned}
\ell_T&=d_i(0)\{1-(1-\delta)^T\}\\
&\quad+\sum_{j=1}^{T-1}\{d_i(j)-d_i(j-1)\}\{1-(1-\delta)^{T-j}\}.
\end{aligned}
\]
Every term is non-decreasing in $\delta$, and, because $0<\delta<1$, it is strictly increasing when its associated positive reliance level or increment is present. Appending a period with weakly higher positive reliance also strictly raises $\ell_T$. For the $\kappa$ comparative static, hold the initial full-use state $\bar\tau$ fixed as $\delta$ changes (equivalently, adjust the observation flow so the pre-adoption steady state is unchanged). Because $\kappa=\delta/(1-\delta)$ is increasing in $\delta$, $s_T=s_0/(1-\ell_T)$ is then increasing in $\kappa$ and in such extensions of $T$.

It remains to carry the state change into both fold parameters. For $j\in\{E,X\}$, let $(\phi_j(s),c_j(s))$ solve $\Phi(\phi,c,s)=0$ and $\Phi_\phi(\phi,c,s)=0$ along the entry or exit fold path. The scalar $s$ is the outside-option noise of the marginal institution at that fold; using one state is a common-cohort reduced form. The adoption menu $(d_H,d_A)=(0,\Delta)$ is held fixed in this comparison. Along either path,
\[
\Delta U_c=\Delta^2(2\phi_j-1)>0,
\qquad
\Delta U_s=\frac{\gamma\chi_H}{2}>0.
\]
Since $g>0$, $\Phi_c>0$ and $\Phi_s>0$. The Jacobian of $(\Phi,\Phi_\phi)$ with respect to $(c,\phi)$ has determinant $\Phi_c\Phi_{\phi\phi}\ne0$, so the implicit-function theorem gives locally unique differentiable fold paths. Differentiating the first fold equation and using $\Phi_\phi=0$ yields
\[
\frac{dc_j(s)}{ds}=-\frac{\Phi_s}{\Phi_c}
=-\frac{\Delta U_s}{\Delta U_c}
=-\frac{\gamma\chi_H}{2\Delta^2(2\phi_j(s)-1)}<0.
\]
Thus, for each $j\in\{E,X\}$,
\[
c_j(s_0)-c_j(s_T)
=\int_{s_0}^{s_T}\frac{\Delta U_s}{\Delta U_c}\,ds>0.
\]
The relevant historical comparison is not the width of two static folds evaluated at the same $s$. It is entry with intact skill and exit with degraded skill:
\[
\{c_E(s_0)-c_X(s_T)\}-\{c_E(s_0)-c_X(s_0)\}
=c_X(s_0)-c_X(s_T)>0.
\]
Since $s_T$ increases with $\kappa$ and with the stated extensions of $T$, each fold shift and the additional historical gap increase while the same fold paths and sign conditions continue to hold. No sign follows for $\beta$, because it moves the financial payoff block and equilibrium reliance as well as the fold locations. Finally, after reliance is set to zero, $\ell_{T+h}=(1-\delta)^h\ell_T>0$ at every finite $h$ and converges to zero; precision therefore returns only asymptotically to its initial steady state.

For the illustrative numbers, we solve both fold equations at $s_0$ and $s_T$ and verify the implicit-function slope by finite differences. At $T=\propHorizon{}$, the exit shift is \propAddedExitShift{}, or \propAddedExitSharePct{} per cent of the intact-skill width \propStaticLoopWidth{}. With constant reliance $d=\Delta$, $s_T/s_0<1/(1-\Delta)$ for every finite horizon and $s_T\to s_{\max}=s_0/(1-\Delta)$ as $T\to\infty$. The limiting exit shift is \propAddedExitShiftLimit{}, or \propAddedExitShareBoundPct{} per cent of the intact-skill width, while the same fold paths continue. At the \propHorizon{}-period endpoint, evaluating both folds at $s_T$ gives width \propTerminalLoopWidth{}, a same-state narrowing of \propSameStateNarrowPct{} per cent relative to \propStaticLoopWidth{}. \qed

\end{document}